\documentclass[letterpaper, 10 pt, conference]{ieeeconf}  % Comment this line out
\IEEEoverridecommandlockouts
\usepackage{balance}
\usepackage{color}
\usepackage{caption}
\usepackage{subcaption}
\usepackage{cite}
\usepackage{empheq}
\usepackage{mathrsfs}

\makeatletter
\let\proof\@undefined                        % undefine \proof
\let\endproof\@undefined                  % undefine \endproof
\makeatother
\usepackage{graphicx,amssymb,amstext,amsmath,amsthm}

\usepackage[bookmarks=true]{hyperref}
\usepackage{algorithm,algorithmicx,algpseudocode}
\algnewcommand{\algorithmicgoto}{\textbf{go to}}%
\algnewcommand{\Goto}[1]{\algorithmicgoto~\ref{#1}}%
\algnewcommand{\LineComment}[1]{\Statex \(\triangleright\) #1}
\algnewcommand{\LineCommentN}[1]{\Statex \hspace{1cm}\(\triangleright\) #1}

\usepackage{multirow}
\usepackage{stfloats}

\newtheorem{prop}{Proposition} % this could go into the preamble
\newtheorem{cor}{Corollary}
\newtheorem{thm}{Theorem}
	\newtheorem{assumption}{Assumption}
\newtheorem{lem}{Lemma}
\newtheorem{defn}{Definition}
\newtheorem{rem}{Remark}
\newtheorem{problem}{Problem}

\newcommand{\yong}[1]{{\color{black} #1}}
\newcommand{\moh}[1]{{\color{black} #1}}

\begin{document}

% paper title
\title{\LARGE \bf Simultaneous Mode, Input and State Set-Valued Observers %for Hidden Mode Switched Linear Systems 
with Applications to Resilient Estimation \yong{a}gainst Sparse Attacks}

% You will get a Paper-ID when submitting a pdf file to the conference system
%\author{Author Names Omitted for Anonymous Review. Paper-ID Sze Zheng Yong}
\author{%
Mohammad Khajenejad \quad\quad Sze Zheng Yong \\%\quad\quad  %Emilio Frazzoli$^{\rm a}$\\
\thanks{%$^1$ The authors are with the Laboratory for Information and Decision Systems,
%Massachusetts Institute of Technology, Cambridge, MA, USA (e-mail: szyong@mit.edu, mhzhu@mit.edu, frazzoli@mit.edu).
$^{\rm a}$ Mohammad Khajenejad and Sze Zheng Yong are with the School for Engineering of Matter, Transport and Energy, Arizona State University, Tempe, AZ, USA (e-mail:  mkhajene@asu.edu, szyong@asu.edu).}
%\thanks{
%$^{\rm b}$ M. Zhu is with the Department of Electrical Engineering, Pennsylvania State University, 201 Old Main, University Park, PA 16802, USA (e-mail: muz16@psu.edu). }%This work was done when M. Zhu was with the Laboratory for Information and Decision Systems at Massachusetts Institute of Technology.}
\vspace{-0.4cm}
%}
}
\maketitle
\thispagestyle{empty}
\pagestyle{empty}

\begin{abstract}
A simultaneous mode, input and state set-valued observer is proposed for hidden mode switched linear systems with bounded-norm noise and %completely 
unknown input signals. The observer consists of two constituents: (i) a bank of mode-matched observers %, one corresponds to each mode, as well as 
and (ii) a mode estimator. Each %of the 
mode-matched observer %corresponds  %The time-varying dynamic can be presented by convex combination of strongly detectable linear systems. The observer 
recursively outputs the mode-matched sets of compatible states and unknown inputs, % corresponding to its mode, 
while the mode estimator eliminates incompatible modes, using a residual-based criterion. Then, the estimated sets of states and unknown inputs are the union of the mode-matched estimates over all compatible modes. Moreover, %multiple 
sufficient conditions to guarantee the elimination of all false modes are provided and % when applying the mode elimination approach. % 		which at each time step, removes the inconsistent modes and their corresponding observers from a bank of filters. The bank of filters, consists of mode correspondent observers, where each of them corresponds to one specific mode, which if is the true mode, the observer simultaneously finds bounded sets of states and unknown inputs that include the true state and inputs. To achieve an effective mode elimination approach, a tractable upper bound signal for the residual's norm is provided, by converting the constraint set of an NP-hard 2-norm maximization problem to a convex set with finite number of extreme points and enumerating the objective function on them. Moreover, multiple sufficient conditions to guarantee that the proposed approach eventually eliminates all the false modes are presented, using matrix lower bound theorem, as well as the convergence behavior of the provided residual norm's upper bound signal. 
		the effectiveness of our approach is exhibited using an illustrative example.
 \end{abstract}
\vspace{-0.1cm}
\section{Introduction}
\vspace{-0.1cm}
Potential vulnerability of Cyber-Physical Systems (CPS) to adversarial attacks and henceforth their security, are emerging as an %extremely 
important and critical issue. Given that %malicious 
attackers are often strategic, there are many potential avenues through which they can cause harm, steal information/power, etc. %, from such systems. 
%The security of Cyber-Physical Systems (CPS) is emerging as an extremely critical and important issue. Since physical and software components are deeply intertwined in CPS,
% has been emerging as a critical issue, since 
%such systems are potentially vulnerable to adversarial attacks, which could be harmful for both the physical systems and their operators. %(\cite{Cardenas.2008b, %Richards.2008, 
	%slay2007lessons,
	%Farwell.2011,ukraine.2016})
	%Hence, a huge effort has been done for   
%	Given that adversarial attackers may behave strategically, there are many potential avenues through which they can achieve their goals of causing harm, information/power theft, etc. 
	%Misleading the system operator by 
	Recent incidents of attacks on CPS, e.g., the Maroochy water system and Ukrainian power grid, %the StuxNet computer worm 	and various industrial security incidents 
	\cite{Cardenas.2008b,ukraine.2016}, % ,Farwell.2011}, 
highlight a need for new resilient estimation and control designs. 
	
	In particular, an adversary's ability to inject counterfeit data into sensor and actuator signals (false data injection) or to compromise an unknown subset of vulnerable sensors and actuators (e.g., \cite{fawzi2014secure,Pasqualetti.2013,pajic2017design,yong2018switching,chong2015observability,shoukry2016event,mo2016secure}) in order to mislead the system operator has been a subject of considerable interest in recent years. This problem can be considered in a more general framework  of hidden mode switched linear systems with unknown inputs and also has applications in urban transportation systems \cite{yong2018switching}, aircraft tracking and fault detection \cite{Liu.Hwang.2011}, etc. 
	
\emph{Literature review.} The filtering problem of hidden mode systems without unknown inputs have been extensively studied
 (see, e.g.,  \cite{Bar-Shalom.2002,Mazor.1998} and references therein). %, 
 %using a multiple model approach. 
 More recently, an extension to consider unknown inputs has been proposed in \cite{yong2018switching} for stochastic systems. However, 
these methods mainly focus on obtaining \emph{point} estimates, i.e., the most likely or best single estimates, and do not directly apply to bounded-error models, 
i.e., uncertain dynamic systems with set-valued uncertainties \yong{(e.g., bounded-norm noise)}, where the sets of all modes, states and unknown inputs that are compatible with sensor observations are desired.

On the other hand, set-membership or set-valued state observers  (e.g., \cite{dahleh1994control,shamma1999set,blanchini2012convex}) are capable of estimating the set of compatible states and are preferable to stochastic estimation when hard accuracy bounds are important, e.g., to guarantee safety. Moreover, a recent extension to also compute the set of unknown input signals in addition to the states has been introduced in \cite{yong2018simultaneous}. However, these %set-membership 
approaches %cannot be used 
do not apply to  hidden mode systems that 
we consider in this paper.

In the context of resilient estimation against sparse false data injection attacks, numerous approaches were proposed (e.g., \cite{fawzi2014secure,Pasqualetti.2013,pajic2017design,yong2018switching,chong2015observability,shoukry2016event,mo2016secure}), but they all only obtain point estimates, as opposed to set-valued estimate\moh{s}. Moreover, only sensor attacks have been considered, although actuator attacks are also a source of concern in CPS security. On the other hand, our prior work in \cite{yong2018simultaneous,khajenejad2019acc} design a fixed-order set-valued observer that simultaneously outputs sets of compatible state and input estimates despite data injection attacks for linear time-invariant and linear parameter-varying systems, without considering the hidden modes, i.e., with the assumption that the subset of attacked sensors and actuators is known.

To consider hidden modes, %it is 
a common approach is to construct \emph{residual} signals, especially for fault detection %purposes 
\cite{patton2013issues}, where %exceeding 
a threshold %by 
based on the residual signal is used to %, 
distinguish between %es 
consistent and inconsistent modes. % of operation.   
Using this idea, % notion, 
\cite{nakahira2018attack} presents a robust control inspired resilient state estimator for %bounded-error models, 
\yong{models with bounded-norm noise that} %, which 
consists of %a combination of a 
local estimators, 
residual detectors and a global fusion detector. %They also provide sufficient conditions to guarantee that their observer is attack resilient. 
However, in their setting, only sensors %can be 
are attacked, while the existence of the observers are assumed with no observer design approach nor performance guarantees. 

 \emph{Contributions.}
 The goal of this paper is to simultaneously consider state and unknown input estimation as well as mode detection for %bounded-error 
 hidden mode switched linear systems with \yong{bounded-norm noise and unknown inputs}. % completely unknown and sparse inputs. % and norm-bounded noise signals.  
 To address this, we propose a multiple-model approach that leverages the optimally designed set-valued state and input $\mathcal{H}_\infty$ observers in our previous work \cite{yong2018simultaneous} to obtain a bank of mode-matched set-valued observers in combination with a novel mode observer based on elimination. Our mode elimination approach uses the upper bound of the norm of %some 
 to-be-designed residual signals to remove %the 
 inconsistent modes from the bank of observers. 
% Each of the observers in the bank, which corresponds to one specific mode of the system, is designed using a recursive fixed-order approach proposed in \cite{yong2018simultaneous}, %observers for linear hidden mode systems with unknown input and bounded noise signals, 
% such that given that its corresponding mode is true, the observer optimally finds bounded sets of states and unknown inputs that contain the true state and unknown input and are compatible/consistent with the measurement outputs, and it is guaranteed that the estimation errors are bounded, assuming that the system is strongly detectable in that mode. %We also provide necessary conditions for the boundedness of the set-valued estimates. 
%Specifically, we consider linear parameter-varying system dynamics that can be presented as a convex combination of linear time-invariant \emph{constituent} dynamics. 
%, known as linear parameter-varying systems, and propose a fixed-order set-valued observer
% and are optimal in the minimum $\mathcal{H}_\infty$-norm sense, i.e., with minimum average power amplification.  %, based on a game-theory approach to $\mathcal{H}_\infty$ estimation. 
% To apply our mode elimination algorithm, 
 In particular, we provide a tractable method to calculate an upper bound signal for the residual's norm
% , by expanding the constraint set of an NP-hard maximization problem and converting it to a 2-norm maximization problem over a hypercube constraint set, where we can find the value function by enumerating the objective function on the extreme points. We also
 and prove that the upper bound signal is a convergent sequence. % to a steady value, so is asymptotically bounded.% and present a novel method for eliminating incompatible modes at each time steps. We show that the upper bound, which can be conveniently calculated by solving a mathematical programming and be used in small time steps, diverges to infinity. So we derive another upper bound, and show that it converges to a steady value, so can be used especially for asymptotic analysis.
  Moreover, we provide sufficient conditions to guarantee that all false modes will be eventually eliminated. % under some reasonable assumptions. %by our mode elimination approach. %, using the convergence property of the upper bound signal and matrix lower bound theorem. 

\emph{Notation.} %We first summarize some notations used throughout the paper. 
$\mathbb{R}^n$ denotes the $n$-dimensional Euclidean space %, $\mathbb{C}$ the field of complex numbers 
and $\mathbb{N}$ nonnegative integers. For a vector $v \in \mathbb{R}^n$ and a matrix $M \in \mathbb{R}^{p \times q}$, $\|v\|_2 \triangleq \sqrt{v^\top v}$, $\| v \|_{\infty} \triangleq \max \limits_{1 \leq i \leq n} v_i$ and $\|M\|_2$ and $\sigma_{\min}(M)$ denote their induced $2$-norm and non-trivial least singular value, respectively. %denotes the 2-norm. %, i.e., $\|v\| \triangleq \sqrt{v^\top v}$, 
\section{Problem Statement} \label{sec:Problem}
%\vspace{-0.1cm}
%\begin{figure}[!h]
%\begin{center}
%\includegraphics[scale=0.2635]{Figures/hybrid_diagram.png}
%\caption{Illustration of switched linear system as a hybrid automaton with two modes, $q$ and $q'$.\label{fig:hyb_diag} }
%\vspace{-0.3cm}
%\end{center}
%\end{figure}
%\noindent\textbf{\emph{System and Unknown Input (or Attack) Assumptions.}} 
Consider a %bounded-error
 hidden mode switched linear system with bounded-norm noise and unknown inputs (i.e., %a dynamical system with multiple modes where 
 a hybrid system with linear and noisy system dynamics in each mode, and %the system dynamics in each mode is linear and uncertain and 
 the mode and some inputs are not known/measured): % (see Figure \ref{fig:hyb_diag}):
\begin{align} \label{eq:sys_desc}
\begin{array}{ll}
\hspace{-0.25cm}%\begin{array}{rl}
%\begin{array}{rl}
%\begin{array}{lll}
 x_{k+1}\hspace{-0.15cm}&=A x_k\hspace{-0.1cm}+\hspace{-0.1cm}B u^{q}_k\hspace{-0.1cm}+\hspace{-0.1cm}G^{q} d^{q}_{k} \hspace{-0.1cm}+\hspace{-0.1cm}w_k, %\nonumber%\\
%& \, \ \ \quad \hspace{4.15cm} 
%x_k\in \mathcal{C}_{q}\\
%(x_k,q)^+ \hspace{-0.15cm}&=(x_k,\delta^{q}(x_k)),  \quad \qquad \quad \qquad \quad \qquad x_k \in \mathcal{D}_{q}\\
\\ %\nonumber
y_k&=C x_k\hspace{-0.1cm} +\hspace{-0.1cm} D u^{q}_k \hspace{-0.1cm}+\hspace{-0.1cm}H^{q} d^{q}_k \hspace{-0.1cm}+\hspace{-0.1cm} v_{k},%q \in \mathbb{Q}. 
\end{array}
\end{align}
%\begin{align} 
%\hspace{-0.25cm}%\begin{array}{rl}
 %x_{k+1}\hspace{-0.15cm}&=A x_k\hspace{-0.1cm}+\hspace{-0.1cm}B u_k\hspace{-0.1cm}+\hspace{-0.1cm}G d_{k} \hspace{-0.1cm}+\hspace{-0.1cm}w_k, \nonumber%\\
%& \, \ \ \quad \hspace{4.15cm} 
%x_k\in \mathcal{C}_{q}\\
%(x_k,q)^+ \hspace{-0.15cm}&=(x_k,\delta^{q}(x_k)),  \quad \qquad \quad \qquad \quad \qquad x_k \in \mathcal{D}_{q}\\
%\\ \label{eq:hybridDyn}
%y_k&=C_{k} x_k\hspace{-0.1cm} +\hspace{-0.1cm} D u_k \hspace{-0.1cm}+\hspace{-0.1cm}H d_k \hspace{-0.1cm}+\hspace{-0.1cm} v_{k} \vspace{-0.4cm}%\end{array}
%\end{align}
%\begin{align} \label{eq:hybridDyn}
%\begin{array}{ll}
%x_{k+1}&=A x_k+B u_k+G d_k + w_k,\\
%y_k&=C x_k +D u_k + H d_k + v_k, \end{array}
%\end{align}
%\begin{align} \label{eq:hybridDyn}
%\nonumber x_{k+1}&=A_k^{q} x_k+B_k^{q} u^{q}_k+G_{k}^{q} d^{q}_{k} +w^{q}_k \\
%\nonumber & \, \ \ \quad \hspace{4.15cm} x_k\in C_{q}\\
%x_k^+&=x_k  \  \quad \qquad \quad \qquad x_k \in D_{q}\\
%\nonumber y_k&=C^{q}_{k} x_k + D^{q}_{k} u^{q}_k +H^{q}_{k} d^{q}_k + v^{q}_{k}
%\end{align}
%\begin{align} \label{eq:hybridDyn}
%\nonumber x_{k+1}&=A_k^{q'} x_k+B_k^{q'} u^{q'}_k+G_{k}^{q'} d^{q'}_{k} +w^{q'}_k \\
%\nonumber & \, \ \ \quad \hspace{4.15cm} x_k\in C_{q'}\\
%x_k^+&=x_k  \  \quad \qquad \quad \qquad x_k \in D_{q'}\\
%\nonumber y_k&=C^{q'}_{k} x_k + D^{q'}_{k} u^{q'}_k +H^{q'}_{k} d^{q'}_k + v^{q'}_{k}
%\end{align}
where $x_k \in \mathbb{R}^n$ is the continuous system state and $q \in \mathbb{Q}=\{1,2,\dots,Q\}$ is the hidden discrete state or \emph{mode}. For each (fixed) mode $q$, 
% $q \in \{1,2,\hdots,N\}$ the hidden discrete state or \emph{mode}. For each mode $q$, 
$u^q_k \in U^q_{k} \subset \mathbb{R}^m$ is the \emph{known} input, $d^q_k \in \mathbb{R}^p$ the unknown but \emph{sparse} input or attack signal, i.e., every vector $d^q_k$ has precisely $\rho \in \mathbb{N}$ nonzero elements where $\rho$ is a known parameter, $y_k \in \mathbb{R}^l$ is the output,
% $\delta^{q}(\cdot)$ the mode transition function, $\mathcal{C}_{q}$ and $\mathcal{D}_{q}$ are flow and jump sets, 
%while the process noise 
whereas $w_k \in \mathbb{R}^n$ and $v_k \in \mathbb{R}^l$ are process and measurement 2-norm bounded disturbances with known parameters $\eta_w$ and $\eta_v$ as their 2-norm bounds respectively.
% are assumed to be mutually uncorrelated, zero-mean, Gaussian white random signals with known covariance matrices, $Q^{q}_k=\mathbb{E} [w_k^{q} w_k^{q \top}] \succeq 0$ and $R^{q}_k=\mathbb{E} [v^{q}_k v_k^{q \top}] \succ 0$, respectively. 
The matrices $A \in \mathbb{R}^{n \times n}$, $B \in \mathbb{R}^{n \times m}$, $G^q \in \mathbb{R}^{n \times p}$, $C \in \mathbb{R}^{l \times n}$, $D \in \mathbb{R}^{l \times m}$ and $H^q \in \mathbb{R}^{l \times p}$ are known
%. $x_0$ is assumed to be independent of $v^{q}_k$ and $w^{q}_k$ for all $k$.
 and no prior `useful' knowledge or assumption of the dynamics of $d^q_k$, except \emph{sparsity} is assumed. 
 
 More precisely, $G^q$ and $H^q$ represent the different hypothesis for each mode $q \in \mathbb{Q}$, about the sparsity pattern of the unknown inputs, which in the context of sparse attacks corresponds to which actuators and sensors are attacked or not attacked. In other words, we assume that $G^q=G\mathbb{I}^q_G$ and $H^q=H\mathbb{I}^q_H$ for some input matrices $G \in \mathbb{R}^{n \times t_a} $ and $H \in \mathbb{R}^{l \times t_s} $, where 
$t_a$ and $t_s$ are the number of vulnerable actuator and sensor signals respectively. Note that $\rho^q_a \leq t_a \leq m$ and  $\rho^q_s \leq t_s \leq l$, where $\rho^q_a$ ($\rho^q_s$) is the number of attacked actuator (sensor) signals and clearly cannot exceed the number of vulnerable actuator (sensor) signals, which in turn cannot exceed the total number of actuators (sensors). Furthermore, we assume that the total number of unknown inputs/attacks in each mode is known and equals $\rho=\rho_a+\rho_s$ (sparsity assumption). Moreover, \moh{the \emph{index matrix}} $\mathbb{I}^q_G \in \mathbb{R}^{t_a \times \rho}$ ($\mathbb{I}^q_H \in \mathbb{R}^{t_s \times \rho}$) represents the sub-vector of $d_k \in \mathbb{R}^{\rho}$ \yong{that indicates} signal magnitude attacks on the actuators (sensors).  
 
 Note that the approach in our paper can be easily extended to handle mode-dependent $A$, $B$, $C$, $D$, $w_k$, $v_k$, $\eta_w$ and $\eta_v$ but is omitted to simplify the notation. 
Moreover, throughout the paper, we assume, without loss of generality, that for each possible mode $q$, the system $(A,G^q,C,H^q)$ is strongly detectable \cite[Definition 1]{yong2018simultaneous}, since this is a necessary and sufficient condition for obtaining meaningful set-valued state and input estimates when the mode is known. % and it is worth mentioning that our approach tries to indicate which sensors and actuators are not attacked. 

%system has
%\emph{strong detectability}, i.e., the initial condition $x_0$ and the unknown input sequence $\{d^{q_j}_j\} ^{r-1}_{j=0}$ can be uniquely determined from the measured output sequence $\{y_i \}^r_{j=0}$ of a sufficient number of observations, i.e.,  $r \geq r_0$ for some $r_0 \in \mathbb{N}$ (see \cite[Section 3.2]{Yong.Zhu.ea.14} for necessary and sufficient conditions for this property) and the required rank condition for the existence of a stable filter \cite[Theorem 9]{Yong.Zhu.ea.14} is satisfied.
Using the modeling framework above, the simultaneous state, unknown input and hidden mode estimation problem %, addressed in this paper, 
is threefold and can be stated as follows:%\newline
%\begin{enumerate}
\begin{problem}
Given a switched linear hidden mode discrete-time bounded-error system with unknown inputs \eqref{eq:sys_desc}, 
\begin{enumerate}[\labelindent=0pt]
\item Design a bank of mode-matched observers that for each mode %, one of the observers 
optimally finds the set estimates  %estimates 
of compatible states and unknown inputs in the minimum $\mathcal{H}_\infty$-norm sense, i.e., with minimum average power amplification, conditional on the mode being true.
\item Develop a mode observer via elimination and the corresponding criterion to eliminate false modes. %as well as its corresponding criterion to detect and eliminate false modes with certainty.
\item Find sufficient conditions for eliminating all false modes.
\end{enumerate}
\end{problem}
 %Moreover, in the time-invariant case, derive necessary and sufficient conditions for convergence of filter gains to a stabilizing/steady-state solution.}
%\subsection{Problem Statement}
%\vspace{-0.1cm}
%\begin{problem}
%Develop a criteria to detect and eliminate false modes.
%\end{problem} 
%% incompatible modes, given the observation and the above bank of observers.}
%\begin{problem}
%Provide sufficient conditions for asymptotically elimination of all the false modes.
%\end{problem}
% an attack-resilient set-valued observer for system \eqref{eq:system} %, i.e., a state filter 
%that computes a bounded set of state estimates that contains the true state and identifies the set of compatible attack signals irrespective of the magnitude of false data injection attacks on its actuators and sensors. In addition, recommend preventative attack mitigation strategies based on detectability conditions.}% characterize fundamental limitations to attack resilience: (i) the maximum number of asymptotically correctable signal attacks and (ii) the maximum number of required models with this estimator. }
%\end{enumerate}
%The objective of this paper\footnote{\yong{Due to space limitation, a technical characterization of the inference algorithm will be presented in an upcoming companion paper\cite{Yong.Zhu.ea.ACC15b}.}} is to design an optimal recursive filter algorithm which simultaneously estimates the system state $x_k$, the unknown input $d^{q}_k$ and the hidden mode $q$ based on %an initial estimate $\hat{x}_0$ and 
%the measurements up to time $k$, $\{y_0,y_1,\hdots, y_k \}$. 
\section{Proposed Observer Design}
In this section, we propose a multiple-model approach for simultaneous mode, state and unknown input estimation for \eqref{eq:sys_desc}, where the goal of the observer is to find compatible set estimates $\hat{D}_k$, $\hat{X}_k$ and $\hat{\mathbb{Q}}_k$ for unknown inputs, states and modes at time step $k$, respectively.
%\vspace{-0.1cm}
\subsection{Overview of Multiple-Model Approach} \label{sec:prelim}
%\vspace{-0.1cm}
 The multiple-model design approach consists of three components: (i) \yong{d}esigning a bank of mode-matched set-valued observers, (ii) designing a mode observer for eliminating incompatible modes using residual detectors, and (iii) a global fusion observer that outputs the desired set-valued mode, input and state estimates. 
%In this section, we present a brief summary of the set-valued optimal filter for linear systems with unknown inputs. For detailed proof and derivation of the observer, the reader is referred to \cite{yong2018simultaneous}.
% Moreover, we define a \emph{generalized innovation} and show that it is a Gaussian white noise. These form an essential part of the multiple model estimation algorithm that  we will describe in Section \ref{sec:MainResult}. The algorithm runs a bank of $\mathfrak{N}$ filters (one for each mode) in parallel and each of the the filter are in essence the same except for the different sets of matrices and signals $\{A_k^{q},B_k^{q},C_k^{q},D_k^{q},G_k^{q},H_k^{q},q^{q},R_k^{q},u_k^{q},d_k^{q}\}$. Hence, to simplify notation, the conditioning on the mode $q$ is omitted in the entire Section \ref{sec:prelim}.
%\vspace{-0.1cm}
\subsubsection{Mode-Matched Set-Valued Observer} \label{sec:ULISE}
First, we design a bank of mode-matched observers, \yong{which} consists of $Q$ simultaneous state and input $\mathcal{H}_\infty$ set-valued observers based on the optimal fixed-order observer \moh{design} in \cite{yong2018simultaneous}, which we briefly summarize here. %for brevity, we present a brief summary of it. % For detailed proof and derivation of the observer, the reader is referred to \cite{yong2018simultaneous}.
%\vspace{-0.1cm}
For each mode-matched observer corresponding to mode $q$, %Suppose that the specific mode $q$ at time step $k$ is the true mode. 
following the approach in \cite[Section 3.1]{yong2018simultaneous}, we consider set-valued fixed-order estimates of the form:
\begin{align}
\hat{D}^{q}_{k-1}&=\{d_{k-1} \in \mathbb{R}^p: \|d_{k-1}-\hat{d}^{q}_{k-1}\|\leq \delta^{d,q}_{k-1}\},\\
%\hat{X}^{\star,q}_k&=\{x_k \in \mathbb{R}^n: \|x_k-\hat{x}^{\star,q}_{k|k}\|\leq \delta^{x,\star,q}_{k}\},\\
\hat{X}^{q}_k&=\{x_k \in \mathbb{R}^n: \|x_k-\hat{x}^{q}_{k|k}\| \leq \delta^{x,q}_k\},
\end{align}
where their centroids %can be 
are obtained with 
% Then, the system, after a similarity transformation $T^q =\begin{bmatrix} T^{q \top}_1 & T^{q\top}_2 \end{bmatrix}^\top \triangleq U^{q\top} =\begin{bmatrix} U^q_{1} & U^q_{2} \end{bmatrix}^\top $  is given by:
%\begin{align}
%x_{k+1} & = A x_k+B u^{q}_k+G^{q}_{1} d^{q}_{1,k} +G^{q}_{2} d^{q}_{2,k} +w_k  \label{eq:sysX}\\
%z^{q}_{1,k}&= C^{q}_{1} x^{q}_k + D^{q}_{1} u^{q}_k +\Sigma^{q} d^{q}_{1} + v^{q}_{1,k} \label{eq:z1}\\
%z^{q}_{2,k} &= C^{q}_{2} x^{q}_k + D^{q}_{2,k} u_k + v^{q}_{2,k} \label{eq:z2},
%\end{align}
%where using singular value decomposition, $H^q=\begin{bmatrix}U^q_{1}& U^q_{2} \end{bmatrix} \begin{bmatrix} \Sigma^q & 0 \\ 0 & 0 \end{bmatrix} \begin{bmatrix} V^{q{\, \top}}_1 \\ V^{q{\, \top}}_2 \end{bmatrix}$, $G^q_{1} \triangleq G^q V^q_{1}$, $G^q_{2} \triangleq G^q V^q_{2}$, $H^q_{1} \triangleq H^q V^q_{1}=U^q_{1}\Sigma^q$, $C^q_{1} \triangleq U^{q\top}_1 C$, $C^q_{2} \triangleq U^{q\top}_2 C$, $D^q_{1} \triangleq U^{q\top}_1 D$, $D^q_{2} \triangleq U^{q\top}_2 D$, $v^q_{1,k} \triangleq U^{q\top}_1 v_k$ and $v^q_{2,k} \triangleq  U^{q\top}_2 v_k$.The transformation essentially decomposes the unknown input $d^q_k$ and the measurement $y_k$ each into two orthogonal components, i.e., $d^{q}_{1,k} \in \mathbb{R}^{p_{H^{q}}}$ and $d^{q}_{2,k} \in \mathbb{R}^{p-p_{H^{q}}}$; as well as $z^{q}_{1,k}\in \mathbb{R}^{p_{H^{q}}}$ and $z^{q}_{2,k} \in \mathbb{R}^{l-p_{H^{q}}}$, where $p_{H^{q}}=\textrm{rank}(H^{q})$. 
%Given measurements up to time $k-1$, 
the following three-step recursive observer that is optimal in $\mathcal{H}_{\infty}$-norm sense: %can be summarized as follows:

\noindent\emph{Unknown Input Estimation}:
\begin{align}
\begin{array}{rl}
\hat{d}^{q}_{1,k} \hspace{-0.2cm}&=M^{q}_{1} (z^{q}_{1,k}-C^{q}_{1} \hat{x}^{q}_{k|k}-D^{q}_{1} u^{q}_k)\\% \label{eq:variant1} \\
\hat{d}^{q}_{2,k-1}\hspace{-0.2cm}&=M^{q}_{2} (z^{q}_{2,k}-C^{q}_{2} \hat{x}^{q}_{k|k-1}-D^{q}_{2} u^{q}_k)\\% \label{eq:d2}\\
\hat{d}^{q}_{k-1}\hspace{-0.2cm}&= V^{q}_{1} \hat{d}^{q}_{1,k-1} + V^{q}_{2} \hat{d}^{q}_{2,k-1} \end{array}%\label{eq:d}
\end{align}
\emph{Time Update}:
\begin{align}
\hspace{-0.3cm}\begin{array}{rl}
 \hat{x}^{q}_{k|k-1}\hspace{-0.1cm}&=A \hat{x}^{q}_{k-1 | k-1} + B u^{q}_{k-1} + G^{q}_{1} \hat{d}^{q}_{1,k-1} \\%\label{eq:time} \\
\hat{x}^{\star,q}_{k|k}&=\hat{x}^{q}_{k|k-1}+G^{q}_{2} \hat{d}^{q}_{2,k-1} %\label{eq:xstar}
\end{array}\hspace{-0.3cm}
\end{align}
\emph{Measurement Update}:
\begin{align}
\hat{x}^{q}_{k|k}
&= \hat{x}^{\star,q}_{k|k} +\tilde{L}^{q}  (z^{q}_{2,k}-C^{q}_{2} \hat{x}^{\star,q}_{k|k}-D^{q}_{2} u^{q}_k)  \quad \label{eq:stateEst}
\end{align}
%where $\hat{x}_{k-1|k-1}$, $\hat{d}_{1,k-1}$, $\hat{d}_{2,k-1}$ and $\hat{d}_{k-1}$ denote the optimal estimates of $x_{k-1}$, $d_{1,k-1}$, ${d}_{2,k-1}$ and $d_{k-1}$; $\Gamma_k \in \mathbb{R}^{p_{\tilde{R}} \times l-p_{H_k}}$ is a design matrix that is chosen to project the ``innovation" $\overline{\nu}_k:=z_{2,k}-C_{2,k} \hat{x}^\star_{k|k}-D_{2,k} u_k$ onto a vector of $p_{\tilde{R}}$ independent random variables, while
%$\overline{L}_k \in \mathbb{R}^{n \times p_{\tilde{R}}}$, $M_{1,k} \in \mathbb{R}^{p_{H_k} \times p_{H_k}}$ and $M_{2,k} \in \mathbb{R}^{(p-p_{H_k}) \times (l-p_{H_k})}$ are filter gain matrices that minimize the state and input error covariances. For the sake of completeness, the optimal input and state filter in \cite{Yong.Zhu.ea.14} is reproduced in Algorithm \ref{algorithm1}.
where %$L \in \mathbb{R}^{n \times l}$, 
$\tilde{L}^q \in \mathbb{R}^{n \times (l-p_{H^q})}$, $M^q_{1} \in \mathbb{R}^{p_{H^q} \times p_{H^q}}$ and $M^q_{2} \in \mathbb{R}^{(p-p_{H^q}) \times (l-p_{H^q})}$ are observer gain matrices that are chosen in %Theorem \ref{thm:filterbanks}
the following theorem from \cite{yong2018simultaneous}  to minimize the ``volume'' of the set of compatible states and unknown inputs, quantified by the radii $\delta^{d,q}_{k-1}$ and $\delta_k^{x,q}$.

\begin{thm}\cite[Lemma 2 \& Theorem 4]{yong2018simultaneous} \label{thm:filterbanks}
Suppose the system $(A,G^q,C,H^q)$ is strongly detectable, $M^q_{1} \Sigma^q=I$ and $M^q_{2}C^q_{2} G^q_{2}=I$. Then, for each mode $q$, there exists a stable and optimal (in $\mathcal{H}_{\infty}$-norm sense) observer with gain $\tilde{L}^q$, where the input and state estimation errors, $\tilde{d}^q_{k-1}\triangleq d^q_{k-1}-\hat{d}^q_{k-1}$ and %, $\hat{x}^\star_{k|k}$ and 
$\tilde{x}^q_{k|k}\triangleq x_k- \hat{x}^q_{k|k}$, are bounded for all $k$ (i.e., 
the set-valued estimates are bounded with radii $\delta^{d,q}_{k-1}, \delta_k^{x,q} < \infty$), and the observer gains and the set estimates are given in  \cite[Theorem 2 \& Algorithm 1]{yong2018simultaneous}. 
\end{thm}
\subsubsection{Mode Estimation Observer}
To estimate the set of compatible modes, we consider an elimination approach that compares residual signals against some thresholds. Specifically, we will eliminate a specific mode $q$, if $\|r^{q}_k\|_2>\hat{\delta}^{q}_{r,k}$, where the residual signal $r^{q}_k$ is defined as follows and the thresholds $\hat{\delta}^{q}_{r,k}$ will be derived in Section  \ref{sec:MainResult}.
%The mode estimation observer, which is summarized in Algorithm \ref{algorithm1}, eliminates incompatible modes, using \emph{residual} signals.
\vspace{-0.1cm}
\begin{defn} [Residuals] \label{defn:computedresidual}
%Suppose that $q$ is a compatible mode with some generic upper bound 
For each mode $q$ at time step $k$, % considering computed observation $z^q_{2,k}$, for each attack mode $q$, 
 the residual signal is defined as:
\begin{align}
\nonumber r^{q}_k \triangleq z^{q}_{2,k}-C^{q}_{2} \hat{x}^{\star,q}_{k|k}-D^{q}_{2} u^{q}_k.%=C^q_2 \tilde{x}^{\star,q}_{k|k}+v^q_{2,k}=\mathbb{A}^q_k {t}_k
\end{align}
\end{defn}
%Algorithm \ref{algorithm1} eliminates a specific mode $q$ at time $k$, if $\|r^{q}_k\|_2>\hat{\delta}^{q}_{r,k}$ based on Theorem \ref{thm:online_mode_elimination}, which will be provided in Section \ref{sec:MainResult}, as well as the formal definition and how to compute $\hat{\delta}^{q}_{r,k}$.
\subsubsection{Global Fusion Observer}
Then, combining the outputs of both components above, our proposed global fusion observer will provide
 mode, unknown input and state set-valued estimates at each time step $k$  as:
\begin{align*}
\begin{array}{c}
\hat{\mathbb{Q}}_k=\{q \in \mathbb{Q} \ \vline \ \|r^q_k\|_2 \leq \hat{\delta}^q_{r,k} \},
\\ \hat{D}_{k-1}=\cup_{q \in \hat{\mathbb{Q}}_k} D^q_{k-1},\  \hat{X}_k=\cup_{q \in \hat{\mathbb{Q}}_k} X^q_{k}.
\end{array}
\end{align*} 

\noindent The %above 
multiple-model approach is summarized 
 in Algorithm \ref{algorithm1}.

%%%%%%%%%%%%%%%%%%%%%%%%%%%%%%%%%%%%%%%%%%%%%%%%%%%%%%%%%%%%%%%%%%%%%%%%%%%%%%%
 \begin{algorithm}[!t] \small
\caption{Simultaneous Mode, State and Input Estimation }\label{algorithm1}
\begin{algorithmic}[1]
  \State $\hat{\mathbb{Q}}_0=\mathbb{Q}$;
  \For {$k =1$ to $N$}
  \For {$q \in \hat{\mathbb{Q}}_{k-1}$}
  \LineComment{Mode-Matched State and Input Set-Valued Estimates}
 \Statex \hspace{0.4cm} Compute $T^q_2,M^q_{1},M^q_{2},\tilde{L}^q,\hat{x}^{\star,q}_{k|k},\hat{X}^{q}_{k},\hat{D}^q_{k-1}$ via Theorem \ref{thm:filterbanks};%$\mathbb{Q}_k=\emptyset$ 
\Statex \hspace{0.4cm} $z^q_{2,k}=T^q_2y_k$;
 \LineComment{Mode Observer via Elimination}
 \Statex \hspace{0.4cm} $\hat{\mathbb{Q}}_k=\hat{\mathbb{Q}}_{k-1}$;
% \Statex \hspace{0.5cm} Compute $\hat{x}^{\star,q}$ via Theorem \ref{} and set $z^q_{2,k}=T^q_2y_k$;  
 \Statex \hspace{0.4cm} Compute $r^q_k$ via Definition \ref{defn:computedresidual} and $\hat{\delta}^q_{r,k}$ via Theorem \ref{thm:resid_comp_up_bound};
 %\Statex \hspace{0.5cm} Compute $\hat{\delta}^q_{r,k}$ via Theorem \ref{};
 %\Statex \hspace{0.5cm}
  \If {$\|r^q_k\|_2>\hat{\delta}^q_{r,k}$} $\hat{\mathbb{Q}}_k=\hat{\mathbb{Q}}_{k} \backslash \{q\}$;% \Else \  $\hat{\mathbb{Q}}_k=\{q\} \cup \hat{\mathbb{Q}}_k$;
  %\Else \  $\mathbb{Q}_k=\{q\} \cup \mathbb{Q}_k$; 
 \EndIf
 \EndFor
 \LineComment{State and Input Estimates}
 \State $\hat{X}_k=\cup_{q \in \hat{\mathbb{Q}}_k} \hat{X}^q_k$; \ $\hat{D}_k=\cup_{q \in \hat{\mathbb{Q}}_k} \hat{D}^q_k$;
 \EndFor
\end{algorithmic}
\end{algorithm}
%%%%%%%%%%%%%%%%%%%%%%%%%%%%%%%%%%%%%%%%%%%%%%%%%%%%%%%%%%%%%%%%%%%%%%%%%%%%%%%%%%%%
\vspace{-0.1cm}
\subsection{Mode Elimination Approach} \label{sec:MainResult}
\vspace{-0.05cm}
The idea is simple. If the residual signal of a particular mode exceeds its upper bound conditioned on this mode being true, we can conclusively rule it out as incompatible.
To do so, for each mode $q$, we first compute an upper bound ($\hat{\delta}^q_{r,k}$) for the 2-norm of its corresponding residual at time $k$, conditioned on $q$ being the \emph{true} mode. Then, comparing the 2-norm of residual signal in Definition \ref{defn:computedresidual} with $\hat{\delta}^q_{r,k}$, we can eliminate mode $q$ if the residual's 2-norm is strictly greater than the upper bound. 
%, by Theorem \ref{thm:online_mode_elimination}, we infer that we can eliminate mode $q$ with certainty, i.e., $q$ is \emph{not} the true mode at time $k$. 
This can be formalized using the following proposition and theorem. %o assert it more formally, we first need the following proposition.
\begin{prop} \label{prop:residecomposition}
Consider mode $q$ at time step $k$, its residual signal $r^q_k$ (as defined in Definition \ref{defn:computedresidual}) and the unknown true mode $q^{*}$. Then, %$r^q_k=r^{q|*}_k+

\vspace{-0.35cm}
\begin{small}
\begin{align}
%\nonumber &r^q_k=r^{q|*}_k+\Delta r^{q|q*}_k, 
\nonumber &r^q_k=r^{q|*}_k+\Delta r^{q|q*}_k,  \textstyle{where}
\\ \nonumber &r^{q|*}_k \triangleq z^{q*}_{2,k}-C^q_2 \hat{x}^{\star,q}_{k|k}-D^q_2 u^q_{k}=T^{q*}_2y_k-C^q_2\hat{x}^{\star,q}_{k|k}-D^q_2 u^q_{k},
\\ \nonumber &\Delta r^{q|q*}_k \triangleq (T^q_2-T^{q*}_2)y_k,
 \end{align}
 \end{small}
% \vspace{-0.1cm}
 \yong{\!\!\!\! where $r^{q|*}_k$ %can be interpreted as 
is the true mode's residual signal (i.e., $q=q^*$), %and
 %\begin{align}
 %\nonumber and \ \Delta r^{q|q*}_k \triangleq (T^q_2-T^{q*}_2)y_k,
 %\end{align}
and $\Delta r^{q|q^*}_k$ %can be considered a 
is the \emph{residual error}.}

 \end{prop}
  \vspace{-0.2cm}
 \begin{proof}
This follows directly from plugging the above expressions into the right hand side term of Definition \ref{defn:computedresidual}.
 \end{proof}
 \vspace{-0.3cm}
 \begin{thm} \label{thm:online_mode_elimination}
 Consider mode $q$ and its residual signal $r^q_k$ at time step $k$.
 Assume that $\delta^{q,*}_{r,k}$ is any signal that satisfies $\| r^{q|*}_k\|_2 \leq \delta^{q,*}_{r,k}$, % \ \forall k $, %, i.e., $\delta^{q,*}_{r,k}$ is an upper envelop signal for $\| r^{q|*}_k\|_2$,
  where $r^{q|*}_k$ is defined in Proposition \ref{prop:residecomposition}. Then, mode $q$ is not the true mode, i.e., can be eliminated at time $k$, if $\| r^q_k \|_2 > \delta^{q,*}_{r,k}.$
 \end{thm}
 \vspace{-0.2cm}
 \begin{proof}%[Proof of Theorem \ref{thm:online_mode_elimination}]
 To use contradiction, suppose $q$ is the true mode. By uniqueness of the true mode $q=q^*$, so $T^q_2=T^{q*}_2$ and by Proposition \ref{prop:residecomposition}, $\Delta r^{q|q*}_k=0$ and hence $\|r^q_k \|_2=\|r^{q|*}_k \|_2 \leq \delta^{q,*}_{r,k}$, which contradicts with the assumption. % in theorem's statement. This completes the proof.
 \end{proof}
% \vspace{-0.3cm}
%  \begin{rem}
%$r^{q|*}_k$ %can be interpreted as 
%\yong{is} the true mode's residual signal (i.e., $q=q^*$), %and
% %\begin{align}
% %\nonumber and \ \Delta r^{q|q*}_k \triangleq (T^q_2-T^{q*}_2)y_k,
% %\end{align}
%while $\Delta r^{q|q^*}_k$ %can be considered a 
%\yong{is the} \emph{residual error}.
% \end{rem}
 
  \vspace{-0.1cm}
 \subsection{Tractable Computation of Thresholds} % $\hat{\delta}^q_{r,k}$}
\vspace{-0.05cm}
Theorem \ref{thm:online_mode_elimination} provides %us 
a sufficient condition for mode elimination at each time step. %In order to be able 
To apply this sufficient condition, we need to compute an % \emph{bounded} 
upper bound for $\| r^{q|*}_k\|_2$, i.e., our $\delta^{q,*}_{r,k}$ signal (\yong{cf.} Theorem \ref{thm:resid_comp_up_bound}) and show that it is bounded in the following lemmas. 
 %We do this task in three steps. First, we find an expression for $ r^{q|*}_k$ in terms of noise signals and initial state through Lemma \ref{lem:resdef}, where we also derive an expression for $ r^{q}_k$ in terms of noise signals, initial state and unknown inputs. 
 %\color{red}For this purpose, we primarily need to find an expression for $ r^{q|*}_k$ in terms of noise signals and initial state and to prove that $\| r^{q|*}_k\|_2$ at each time step is bounded through the following Lemmas \ref{lem:resdef} and \ref{lem:existance}, respectively. \color{black}% in the following theorem.
 %Note that we also derive an expression for $ r^{q}_k$ in terms of noise signals, initial state and unknown inputs, that will be used later in . 
% and finally, we compute an upper bound for it via Theorem \ref{thm:resid_comp_up_bound}.  
 \vspace{-0.2cm}
\begin{lem} \label{lem:resdef}
Consider any mode $q$ with the unknown true mode being $q^{*}$. Then, at time step $k$, we have  %is the true mode, then the residual signal at time step $k$ can be obtained as
\begin{align}
 r^{q|*}_k &= C^q_2 \tilde{x}^{\star,q}_{k|k}+v^q_{2,k}=\mathbb{A}^q_k {t}_k, \label{eq:resid_ideal}
%\\ r^{q}_k &= \begin{bmatrix} \mathbb{T}^{q,q^*}_k & \mathbb{B}^{q,q^*}_k & \mathbb{D}^{q,q^*}_k  \end{bmatrix} \begin{bmatrix} t^{\top}_k & u^{q^*\top}_{0:k} & d^{\top}_{0:k} \end{bmatrix}^{\top}, \label{eq:resid_comp}
 \end{align}
 where\begin{small} ${t}_k \hspace{-0.1cm}\triangleq \hspace{-0.1cm}\begin{bmatrix} \tilde{x}^{\top}_{0|0} & w^{\top}_0 & \dots & w^{\top}_{k-1} & v^{\top}_0 & \hspace{-0.1cm}\dots\hspace{-0.1cm} & v^{\top}_k \end{bmatrix}^{\top} \hspace{-0.2cm}\in \mathbb{R}^{\moh{(n+l)(k+1)}}$, \end{small}
 \vspace{-0.15cm}
 \begin{small}
 \begin{align}
% \nonumber {t}_k \triangleq& \begin{bmatrix} \tilde{x}^{\top}_{0|0} & w^{\top}_0 & \dots & w^{\top}_{k-1} & v^{\top}_0 & \dots & v^{\top}_k \end{bmatrix}^{\top} \in \mathbb{R}^{\moh{(n+l)(k+1)}},\\
% \\ \nonumber &u^{q^*}_{0:k} \triangleq \begin{bmatrix} u^{q*\top}_k & u^{q*\top}_{k-1} & \dots u^{q*\top}_0  \end{bmatrix}^{\top}, d_{0:k} \triangleq \begin{bmatrix} d^{\top}_k & d^{\top}_{k-1} & \dots d^{\top}_0  \end{bmatrix}^{\top}
% \end{align}
%
% \begin{align}
\quad \quad \mathbb{A}^q_k \triangleq&
%\begin{align} 
%\begin{small}
\nonumber  [ C^q_2 \overline{A}^q {A^q_e}^{k-1}  \vline  C^q_2 \overline{A}^q {A^q_e}^{k-2}B^q_{e,w} \vline  C^q_2 \overline{A}^q {A^q_e}^{k-2}B^q_{e,w}  \dots  \\ \nonumber &C^q_2 \overline{A}^q {A^q_e}^{k-1-i}B^q_{e,w}  \vline \dots \vline  C^q_2 \overline{A}^q {A^q_e}B^q_{e,w}  \vline  C^q_2 B^{\star,q}_{e,w}  \vline  \\ \nonumber &C^q_2 \overline{A}^q {A^q_e}^{k-2}B^q_{e,v_1}  \vline  C^q_2 \overline{A}^q {A^q_e}^{k-2}(B^q_{e,v_1}+ {A^q_e}B^q_{e,v_2}) \dots
\\ \nonumber  &C^q_2 \overline{A}^q {A^q_e}^{k-1-i}(B^q_{e,v_1}+ {A^q_e}B^q_{e,v_2})  \vline  \dots  \vline  
 \\ \nonumber &C^q_2 \overline{A}^q {A^q_e}(B^q_{e,v_1}+ {A^q_e}B^q_{e,v_2})  \vline  C^q_2 (B^{q, \star}_{e,v_1}+ \overline{A}^q B^q_{e,v_2})  \vline
 \\ \nonumber &C^q_2 B^{q, \star}_{e,v_2}+T^q_2 ] \in \mathbb{R}^{(l-p_{H^q}) \times \moh{(n+l)(k+1)}},
 %\end{small}
%\end{bmatrix}
\end{align}
 \end{small}
%$\hat{A}^q\triangleq A-G^q_{1}M^q_{1} C^q_{1}$, $\Phi\triangleq I-G^q_2 M^q_2 C^q_2$, 
%\begin{align}
%\nonumber \overline{A}^q&\triangleq (I-G^q_2 M^q_2 C^q_2)(A-G^q_{1}M^q_{1} C^q_{1}), %$V_e\triangleq V_1M_1C_1+V_2M_2C_2\hat{A}$, 
%$A_e\triangleq (I-\tilde{L}C_2)\overline{A}$, 
%\\ \nonumber A^q_e&\triangleq (I-\tilde{L}^qC^q_2)\overline{A}^q, B^{\star,q}_{e,w}\triangleq (I-G^q_2 M^q_2 C^q_2),
% \\ \nonumber B^{\star,q}_{e,v1}&\triangleq -(I-G^q_2 M^q_2 C^q_2) (G^q_1 M^q_1 T^q_1), 
 % \\ \nonumber B^q_{e,w}&\triangleq (I-\tilde{L}^qC^q_2)B^{\star,q}_{e,w}, B^q_{e,v1}\triangleq (I-\tilde{L}^qC^q_2)B^{\star,q}_{e,v1}, 
   %\\ \nonumber B^q_{e,v2}&\triangleq (I-\tilde{L}^qC^q_2)B^{\star,q}_{e,v2}-\tilde{L}^qT^q_2, B^{\star,q}_{e,v2}\triangleq -G^q_2 M^q_2T^q_2,
% \end{align}

\vspace{-0.4cm}\noindent with \begin{small}$\overline{A}^q\triangleq \hspace{-0.1cm}(I\hspace{-0.05cm}-\hspace{-0.05cm}G^q_2 M^q_2 C^q_2)(A-G^q_{1}M^q_{1} C^q_{1})$, \hspace{-0.15cm}$A^q_e\hspace{-0.1cm}\triangleq \hspace{-0.1cm}(I\hspace{-0.05cm}-\hspace{-0.05cm}\tilde{L}^qC^q_2)\overline{A}^q, B^{\star,q}_{e,w}\hspace{-0.1cm}\triangleq \hspace{-0.1cm}(I\hspace{-0.1cm}-\hspace{-0.1cm}G^q_2 M^q_2 C^q_2)$, $B^{\star,q}_{e,v1}\hspace{-0.1cm}\triangleq \hspace{-0.1cm}-(I\hspace{-0.1cm}-\hspace{-0.1cm}G^q_2 M^q_2 C^q_2) (G^q_1 M^q_1 T^q_1)$, \hspace{-0.1cm}$B^q_{e,w}\hspace{-0.1cm}\triangleq \hspace{-0.1cm}(I-\tilde{L}^qC^q_2)B^{\star,q}_{e,w}, B^q_{e,v1}\hspace{-0.1cm}\triangleq\hspace{-0.1cm} (I\hspace{-0.1cm}-\hspace{-0.1cm}\tilde{L}^qC^q_2)B^{\star,q}_{e,v1}$ and $B^q_{e,v2}\hspace{-0.1cm}\triangleq \hspace{-0.1cm}(I\hspace{-0.1cm}-\hspace{-0.1cm}\tilde{L}^qC^q_2)B^{\star,q}_{e,v2}\hspace{-0.1cm}-\hspace{-0.1cm}\tilde{L}^qT^q_2, B^{\star,q}_{e,v2}\hspace{-0.1cm}\triangleq \hspace{-0.1cm}-G^q_2 M^q_2T^q_2$.\end{small} 
\end{lem}
\vspace{-0.2cm}\begin{proof}%[Proof of Lemma \ref{lem:resdef}]%[Proof of Lemma \ref{lem:resdef}]
Considering \eqref{eq:resid_ideal}, the first equality comes from Definition \ref{defn:computedresidual} and $z^{q}_{2,k} = C^{q}_{2} x_k + D^{q}_{2,k} u^q_k + v^{q}_{2,k}$ from \cite{yong2018simultaneous}, assuming that $q$ is the true mode, and the second equality is implied by the first equality and the fact in \cite[Appendix C]{yong2018simultaneous} that %(c.f., \cite[Appendix C]{yong2018simultaneous}) %&\begin{array}{l}
\begin{align}
&\begin{array}{rl} \tilde{x}^{\star,q}_{k|k}&=\overline{A}^q{A^q_e}^{k-1} \tilde{x}_{0|0}+\overline{A}^q{A^q_e}^{k-2}\begin{bmatrix} B^q_{e,w}  B^q_{e,v1} \end{bmatrix}\vec{w}_0 %\begin{bmatrix} w_0\\ v_0 \end{bmatrix}\\
 \\ \nonumber &+B_{e,w}^{\star,q} w_{k-1} + (B_{e,v1}^{\star,q}+\overline{A}^qB^q_{e,v2}) v_{k-1} + B_{e,v2}^{\star,q} v_k
  \\ \nonumber &+\textstyle \sum_{i=1}^{k-2} \overline{A}^q{A^q_e}^{k-1-i} \begin{bmatrix} B^q_{e,w} & B^q_{e,v1}+A^q_e B^q_{e,v2} \end{bmatrix} \vec{w}_i, \end{array}
  \\ \nonumber & \vec{w}_k\triangleq \begin{bmatrix} w_k^\top & v_k^\top\end{bmatrix}^\top.  \qedhere\end{align}
     \end{proof} 
     
\begin{lem} \label{lem:existance}
For each mode $q$ at time step $k$, there exists a generic finite valued upper bound $\delta^{q}_{r,k} < \infty$ for  $ \| r^{q|*}_k \|_2$.
 \end{lem}
  \vspace{-0.2cm}\begin{proof}%[Proof of Lemma \ref{lem:existance}]
 Consider the following optimization problem for \moh{$ \| r^{q|*}_k \|_2$} by leveraging Lemma \ref{lem:resdef}:
% \begin{problem}\label{problem:modeelimination}
 %For $\mathbb{A}^q_k$ and ${t}_k$ as defined in  Lemma \ref{lem:resdef}, solve:
\vspace{-0.05cm} \begin{align} \label{eq:genericupperbopund}
  &\delta^{q}_{r,k} \triangleq \max \limits_{t_k} \| \mathbb{A}^q_k {t}_k \|_2 %,t_k \triangleq [\tilde{x}_{0|0} \ w_0 \ ... \ w_{k-1} \ v_0 \ ... \ v_k] \top
 \\ \nonumber  &s.t. \ t_k = \begin{bmatrix} \tilde{x}^{\top}_{0|0} & w^{\top}_0 & \dots & w^{\top}_{k-1} & v^{\top}_0 & \dots & v^{\top}_k \end{bmatrix}^{\top},  
 \\ \nonumber &\| \tilde{x}_{0|0} \|_2 \leq \delta^x_0, \ \|w_i\|_2 \leq \eta_w, \ \| v_j \|_2 \leq \eta_v,
 \\ \nonumber &i \in \{ 0,...,k-1 \}, \ j \in \{ 0,...,k \}.%\mathbb{A}^q_k \ is \ defined \ in \ Lemma \ref{lem:resdef}.
 \end{align}
 %\end{problem}
 The objective 2-norm function is continuous and the constraint set is an intersection of level sets of lower dimensional norm functions, which is closed and bounded\moh{,} so is compact. Hence, by Weierstrass Theorem \cite[Proposition 2.1.1]{bertsekas2003convex}, the objective function attains its maxima on the constraint set and so a finite-valued upper bound exists. %This completes the proof.  
 \end{proof}
 
%\vspace{0.1cm}
%\begin{bmatrix} w_i\\ v_i \end{bmatrix},
%  \end{align} (cf. \cite[Appendix C]{yong2018simultaneous}).
%  Moreover, \eqref{eq:resid_comp} is implied by applying Proposition \ref{prop:residecomposition}, euation \eqref{eq:resid_ideal} and the closed from expression for the output signal
%  
%  \begin{tiny}
%  \begin{align}
%  \nonumber y_k&=\begin{bmatrix} \begin{bmatrix} (CA^k)^{\top} \\ (CA^{k-1})^{\top} \\ \vdots \\ C^{\top} \\ I \end{bmatrix}^{\top} & \begin{bmatrix} H^{\top} \\ (CG)^{\top} \\ (CAG)^{\top} \\ \vdots \\ (CA^{k-1}G)^{\top} \end{bmatrix}^{\top} 
%  & \begin{bmatrix} D^{\top} \\ (CB)^{\top} \\ (CAB)^{\top} \\ \vdots \\ (CA^{k-1}B)^{\top} \end{bmatrix}^{\top}  \end{bmatrix} \begin{bmatrix} t_k \\ \\ d^{q*}_{0:k} \\ \\ u^{q*}_{0:k} \end{bmatrix}%\begin{bmatrix}
%  \end{align}
%  \end{tiny}
%  which can be derived by using system \eqref{eq:sys_desc}'s equations and induction. This, completes the proof.

%Providing the computed residual's upper bound, consider the following Lemma:

%Now, having a closed-form formulation for $r^{q|*}_k$, we show that there exists a \emph{finite} upper bound for its 2-norm, at each time $k$ , through  Lemma \ref{lem:existance}.
%we pretend that the specific mode $q$ is the true mode, and try to provide an upper bound for its residual signal.
\vspace{-0.2cm}

% \vspace{-0.65cm}
Clearly $\delta^{q}_{r,k}$ in Lemma \ref{lem:existance} is the \emph{tightest} possible residual norm's upper bound %\color{red}(c.f \eqref{eq:genericupperbopund} in Appendix)\color{black}  
and potentially can eliminate the most possible number of modes, so is the best choice if we can calculate it. But, notice that although it was straight forward to show that a finite-valued $\delta^{q}_{r,k}$ exists, but since the optimization problem in Lemma \ref{lem:existance} is a \emph{norm maximization} (not minimization) over the intersection of level sets of lower dimensional norm functions, i.e., a non-concave maximization over intersection of quadratic constraints, it is an NP-hard problem \cite{bodlaender1990computational}. To tackle with this complexity, we provide an over-approximation for $\delta^{q}_{r,k}$ in the following Theorem \ref{thm:resid_comp_up_bound}, which we call $\hat{\delta}^q_{r,k}$.% which are of course -more conservative/wider upper bounds for residuals' norms. 
\begin{thm} \label{thm:resid_comp_up_bound}
Consider mode $q$. At time step $k$, let
\begin{small}
\begin{align}
\nonumber &\hat{\delta}^q_{r,k} \triangleq \min \{ \delta^{q,inf}_{r,k},{\delta}^{q,tri}_{r,k} \}, %is an over-approximation for $\delta^{q*}_{r,k},
 \\ \nonumber &\delta^{q,inf}_{r,k} \triangleq \| \mathbb{A}^q_k {t}^{\star}_k \|_2,
\\  \nonumber &\delta^{q,tri}_{r,k} \triangleq \delta^{x,q}_0 \| C^q_2 \overline{A}^q {A^q_e}^{k-1} \|_2  +\eta_w \| C^q_2 \overline{A}^q {A^q_e}^{k-2} \|_2+%\label{eq:trianglenorm}
 \\ \nonumber &\ \ \textstyle\sum_{i=1}^{k-2} [\eta_w\| C^q_2 \overline{A}^q {A^q_e}^{i}B^q_{e,w} \|_2 \hspace{-0.1cm}+\hspace{-0.1cm}\eta_v \| C^q_2 \overline{A}^q {A^q_e}^{i}(B^q_{e,v_1}\hspace{-0.1cm}+\hspace{-0.1cm} {A^q_e}B^q_{e,v_2}) \|_2] %+\|C^q_2 \overline{A}^q {A^q_e}^{k-2}B^q_{e,v_1} \|_2
 \\ \nonumber &\quad +\eta_v(\|C^q_2 \overline{A}^q {A^q_e}^{k-2}B^q_{e,v_1} \|_2+\|C^q_2 (B^{q, \star}_{e,v_1}+ \overline{A}^q B^q_{e,v_2}) \|_2)
 \\ \nonumber &\quad +\|C^q_2 B^{q, \star}_{e,v_2}+T^q_2 \|_2)+\eta_w \| C^q_2 B^{\star,q}_{e,w} \|_2, 
 \end{align}
 \end{small}
 where ${t}^{\star}_k$ is a vertex of the following hypercube: 

\vspace{-0.3cm}%\noindent where ${t}^{\star}_k$ is a vertex of the following hypercube: 
 \begin{align*}
 \begin{array}{l}
\mathcal{X}^q_k \triangleq \big\{ x \in \mathbb{R}^{\moh{(n+l)(k+1)}} \ \vline \\
| x(i) | \leq \begin{cases} \delta^x_0, 1 \leq i \leq n \\ \eta_w, n+1 \leq i \leq \moh{n(k+1)} \\ \eta_v, \moh{n(k+1)}+1 \leq i \leq \moh{(n+l)(k+1)}  \end{cases}\big\},\end{array}
 \end{align*}
 %\end{align},
  i.e., \\ ${t}^{\star}_k(i) \in \begin{cases} \{-\delta^x_0, \delta^x_0 \}, 1 \leq i \leq n, \\ \{-\eta_w, \eta_w \}, n+1 \leq i \leq \moh{n(k+1)}, \\ \{-\eta_v, \eta_v \}, \moh{n(k+1)}+1 \leq i \leq \moh{(n+l)(k+1)}.  \end{cases} $
  \\ Then, $\hat{\delta}^q_{r,k}$ is an over-approximation for $\delta^{q}_{r,k}$ in Lemma \ref{lem:existance}.
   %and
  %\begin{small}
 %\begin{align}
 %\nonumber &\delta^{q,tri}_{r,k} \triangleq \delta^{x,q}_0 \| C^q_2 \overline{A}^q {A^q_e}^{k-1} \|_2  +\eta_w \| C^q_2 \overline{A}^q {A^q_e}^{k-2} \|_2+%\label{eq:trianglenorm}
 %\\ \nonumber &\sum_{i=1}^{k-2} [\eta_w\| C^q_2 \overline{A}^q {A^q_e}^{i}B^q_{e,w} \|_2 +\eta_v \| C^q_2 \overline{A}^q {A^q_e}^{i}(B^q_{e,v_1}+ {A^q_e}B^q_{e,v_2}) \|_2] %+\|C^q_2 \overline{A}^q {A^q_e}^{k-2}B^q_{e,v_1} \|_2
 %\\ \nonumber &+\eta_v(\|C^q_2 \overline{A}^q {A^q_e}^{k-2}B^q_{e,v_1} \|_2+\|C^q_2 (B^{q, \star}_{e,v_1}+ \overline{A}^q B^q_{e,v_2}) \|_2)
 %\\ \nonumber &+\|C^q_2 B^{q, \star}_{e,v_2}+T^q_2 \|_2)+\eta_w \| C^q_2 B^{\star,q}_{e,w} \|_2 
 %\\ \nonumber &\|C^q_2 \overline{A}^q {A^q_e}^{k-2}B^q_{e,v_1} \|_2 ).
% \end{align}
 %\end{small}
 \end{thm}
  \vspace{-0.3cm}\begin{proof}%[Proof of Theorem \ref{thm:resid_comp_up_bound}] 
 Consider the optimization problem
  \begin{align}
  &\delta^{q,inf}_{r,k} \triangleq \max \limits_{t_k} \| \mathbb{A}^q_k {t}_k \|_2 \label{eq:inf_norm}
%,t_k \triangleq [\tilde{x}_{0|0} \ w_0 \ ... \ w_{k-1} \ v_0 \ ... \ v_k] \top
\\ \nonumber &s.t. \ t_k= \begin{bmatrix} \tilde{x}^{\top}_{0|0} & w^{\top}_0 & \dots & w^{\top}_{k-1} & v^{\top}_0 & \dots & v^{\top}_k \end{bmatrix},
 \\ \nonumber  &\ \ \ \ \  \| \tilde{x}_{0|0} \|_{\infty} \leq \delta^x_0, \ \|w_i\|_{\infty} \leq \eta_w, \ \| v_j \|_{\infty} \leq \eta_v,
 \\ \nonumber &\ \ \ \ \ \forall i \in \{ 0,...,k-1 \}, \ \forall j \in \{ 0,...,k \}. %\ and%\mathbb{A}^q_k \ is \ defined \ in \ Lemma \ref{lem:resdef}.
 \end{align}
 Comparing \eqref{eq:genericupperbopund} and \eqref{eq:inf_norm}, the two problems have the same objective functions, while since $\| . \|_{\infty} \leq \| . \|_{2} $, the constraint set for \eqref{eq:genericupperbopund} is a subset of the one for \eqref{eq:inf_norm}. Hence $\delta^q_{r,k} \leq \delta^{q,inf}_{r,k}$. % Since $\delta^q_{r,k}$ is an upper bound for residual's norm, so is $\delta^{q,inf}_{r,k}$.
 Also, it is easy to see that $\hat{\delta}^q_{r,k} \leq \delta^{q,tri}_{r,k}$, using triangle and sub-multiplicative inequalities.
  Moreover, \eqref{eq:inf_norm} is a \emph{maximization} of a convex objective function over a convex constraint (hypercube $\mathcal{X}^q_k$). By a famous result \cite[Corollary 32.2.1]{rockafellar2015convex}, in such a problem, the objective function attains its maxima on some of the extreme points of the constraint set, which in this case are the vertices of the hypercube $\mathcal{X}^q_k$. %This, completes the proof.
 \end{proof}
 
It can be easily seen as a corollary of Theorem \ref{thm:resid_comp_up_bound} that:
 \begin{cor} \label{cor:verticenorm}
$ \eta^t_{k} \triangleq  \| {t}^{\star}_k \|_2=\sqrt{n {\delta^x_o}^2+kn \eta^2_w+(k+1)l\eta^2_v} $.
 \end{cor}
% \begin{rem}
 Theorem \ref{thm:resid_comp_up_bound} enables us to obtain  an upper bound for $\|r^{q|*}_k\|_2$, by enumerating the objective function in \eqref{eq:inf_norm} at vertices of the hypercube $\mathcal{X}^q_k$ and choosing the largest value as $\delta^{q,inf}_{r,k}$. Moreover, we can easily calculate $\delta^{q,tri}_{r,k}$; then, the upper bound is chosen as %using triangle and sub-multiplicative inequalities and finally choosing 
 the minimum of the two as $\hat{\delta}^q_{r,k}$.
% \end{rem}
  %We should mention that
   \begin{rem}
   Although simulation results indicate that especially in earlier time steps, $\delta^{q,inf}_{r,k}$ may have smaller values than $\delta^{q,tri}_{r,k}$, but if we only consider $\delta^{q,inf}_{r,k}$ as the over-approximation and do not use $\delta^{q,tri}_{r,k}$, then we will face two difficulties.  
 %\end{rem}
% \begin{rem}
 %Although by Theorem \ref{thm:infnormbound} we are conveniently equipped with a tool to obtain an upper bound for the residual's norm and mode elimination, 
% but 
 First, as time increases, the number of required enumerations (i.e., the number of hypercube's vertices which is $2^{\moh{(n+l)(k+1)}}$) increases with an exponential rate. Second and more importantly, as Lemma \ref{lem:delta_inf_diverge} will indicate later, $\delta^{q,inf}_{r,k}$ goes to infinity as time increases, so it will be unlikely to eliminate any mode when the time step is large, i.e., asymptotically speaking, $\delta^{q,inf}_{r,k}$ will be useless. In contrast, again by Lemma \ref{lem:delta_inf_diverge}, $\delta^{q,tri}_{r,k}$ converges to some steady-state value, so it can be always used as an over-approximation for $\delta^{q}_{r,k}$ in the mode elimination process. 
   %by  $\delta^{q,inf}_{r,k}$ will become useless for asymptotic ex-post mode elimination.  
 \end{rem}
%\vspace{-0.2cm}
\section{Mode Detectability}
%\vspace{-0.2cm}
 In addition to the nice properties regarding the stability and boundedness of the mode-matched set estimates of state and input obtained from \cite{yong2018simultaneous}, we now provide some sufficient conditions for the system dynamics, which guarantee that regardless of the observations, after some large enough time steps, %most 1 mode will be compatible,
 \emph{all} the false (i.e., not true) modes can be eliminated, when applying Algorithm \ref{algorithm1}. %\color{red}applying our mode elimination approach\color{black}. 
 To do so, first, we define the concept of mode detectability %and assert two required lemmas 
 as well as some assumptions for deriving our sufficient conditions for mode detectability. % Since if this is the case, we make our conclusion \emph{before} realizing any observation, we call it \emph{ex-ante} elimination of the inconsistent modes.
%First we formally define this statement.
%\begin{rem} \label{rem:exantelimitaion}
%Obviously the aforementioned structural conditions are also sufficient to claim that all the modes except at most one of them can be ex-post $\delta^{q,tri}_{r,k}$-eliminated after some time step onward and so they can be ex-post $\delta^{q,tri}_{r,\infty }$-asymptotically-eliminated (c.f., Definition \ref{defn:compatibility}).  
%\end{rem} 

\begin{defn}[Mode Detectability] \label{defn:strong_mode_detecatble}
System \eqref{eq:sys_desc} 
%with sparse attacks and the set of possible modes $Q$.
 %Let $\delta^q_{r,k}$ be a generic upper bound for the residual's norm for mode $q$ and time step $k$. The system 
is called mode detectable if %at each time step $k>0$, at most one mode is $\delta^{q}_{r,k }$-compatible (c.f., Definition \ref{defn:compatibility}). It is called asymptotically strongly $Q-\delta^q_{r,k}$-mode-detectable
 % after some large enough tiime time 
 there exists a natural number $K>0$, such that for all time steps $k \geq K$, all false modes are eliminated.
 % at most one mode is $\delta^{q}_{r,k }$-consistent.
\end{defn}

%
%\balance
%\section{Ex-ante Mode Elimination}\label{sec:exante}
%So far we have provided an upper bound that enable us to eliminate inconsistent modes after realizing all the observations up to the current time step.
\begin{assumption} \label{assumption:boundedness}
There exist known $R_y, R_x \in \mathbb{R}$ such that $\forall k, y_{k} \in Y \triangleq \{ y \in \mathbb{R}^{l} \vline \ \| y \|_2 \leq R_y \}$ and $x_{k} \in X \triangleq \{ x \in \mathbb{R}^{n} \vline \ \| x \|_2 \leq R_x \}$, i.e., there exist known bounds for the whole observation/measurement and state spaces, respectively. %\footnote{JUSTIFICATION...}.
\end{assumption}
%\begin{assumption} \label{assumption:boundedness2}
%There exists known $R_x \in \mathbb{R}$ such that $\forall k,x_{k} \in X \triangleq \{ x \in \mathbb{R}^{n} \vline \ \| x \|_2 \leq R_x \}$, i.e., there exists a known bound for the whole state space.% \footnote{JUSTIFICATION...}.
%\end{assumption}

%\vspace{-0.3cm}
%\begin{rem}
%When $q$ cannot be eliminated using Algorithm \ref{algorithm1}, then $C^q_2 \hat{x}^{\star,q}_k + D^q_2 u^q_k$ should have the same order of magnitude as $z^q_{2,k}$ since although the residual is intuitively rather small if the noise terms are not adversarial, but we have been affected by adversarial disturbances indirectly in the process of selecting the (\emph{possibly wrong}) mode $q$. If this is the case, then $ \|C^q_2 \hat{x}^{\star,q}_{k|k} + D^q_2 u^q_k - C^{q'}_2 \hat{x}^{\star,q'}_{k|k} - D^{q'}_2 u^{q'}_k\|_2$ is unlikely to be greater than $R'^{q,q'}_y$ and it is even unlikelier for it to be greater than ${\delta}^{q,tri}_{r,k}+{\delta}^{q',tri}_{r,k}+R'^{q,q'}_y$, so sufficient condition \eqref{eq:disjointness} is more likely to not hold. Now, consider the case that $q$ is "not a true mode for sure and can be detected that is not true", i.e., can be eliminated using Algorithm \ref{algorithm1}. In this case, it is not granted that $C^q_2 \hat{x}^{\star,q}_{k|k} + D^q_2 u^q_k$ has the same order of magnitude as $z^q_{2,k}$, and it is exactly why our \emph{prior} partial information on the possible set of observations, helps satisfying the sufficient condition. 
%\end{rem}
%\begin{rem}

\begin{assumption}\label{as:2}
The unknown input/attack signal has an \emph{unlimited energy}, i.e., $\lim_{k\to\infty} \|d^{q*}_{0:k}\|_2 = \infty$, where $d^{q*}_{0:k} \triangleq \begin{bmatrix} d^{q*\top}_k & d^{q*\top}_{k-1} & \dots d^{q*\top}_0  \end{bmatrix}^{\top}$.
\end{assumption}
Note that Assumption \ref{as:2} is not restrictive because otherwise, the unknown input/attack  signal must vanish asymptotically, which means that the true mode (with no unknown inputs) can be inferred asymptotically. 

In order to derive the desired sufficient conditions for mode detectability in Theorem \ref{thm:strong_mode_detect}, we first present the following  Lemmas \ref{lem:delta_inf_diverge}--\ref{lem:resdef2}. For the sake of clarity, the proofs of these results are given in the Appendix.
\begin{lem} \label{lem:delta_inf_diverge}
 For each mode $q$,
 \begin{align}
    &\lim_{k\to\infty} \delta^{q,inf}_{r,k}= \infty. \label{eq:r_inf_diverge}
   \\  &\lim_{k\to\infty} \hat{\delta}^{q}_{r,k}=\lim_{k\to\infty} \delta^{q,tri}_{r,k} \leq \lim_{k\to\infty} \overline{\delta}^{q,tri}_{r,k} = \overline{\delta}^{q,tri}_{r}< \infty ,  \label{eq:r_tri_converge}
     \end{align}
     where
 % \begin{small}
% \begin{align}  %\label{eq:upper_tri}
% \begin{array}{rl}
 $\overline{\delta}^{q,tri}_{r,k} \triangleq \delta^{x,q}_0 \| C^q_2 \overline{A}^q {A^q_e}^{k-1} \|_2  +\eta_w \| C^q_2 \overline{A}^q {A^q_e}^{k-2} \|_2%+
+\eta_w[\| C^q_2 \overline{A}^q {A^q_e } \|_2 \|B^q_{e,w}\|_2 \sum_{i=0}^{k-3} (\| {A^q_e} \|^i_2) + \| C^q_2 B^{\star,q}_{e,w} \|_2]%+
+\eta_v [ \| C^q_2 \overline{A}^q {A^q_e}\|_2 \| B^q_{e,v_1}+ {A^q_e}B^q_{e,v_2} \|_2 \sum_{i=0}^{k-3} \| {A^q_e} \|^i_2]%+
+\eta_v[\|C^q_2 B^{q, \star}_{e,v_2}+T^q_2 \|_2+ \|C^q_2 (B^{q, \star}_{e,v_1}+ \overline{A}^q B^q_{e,v_2}) \|_2]%+
+\eta_v\|C^q_2 \overline{A}^q {A^q_e}^{k-2}B^q_{e,v_1} \|_2$, \moh{$\overline{\delta}^{q,tri}_{r} \triangleq \eta_w [\| C^q_2 B^{q,\star}_{e,w} \|_2 +\| C^q_2 \overline{A}^q A^q_e\|_2 /(1-\theta^q)+\| B^q_{e,w}\|_2]+ \eta_v [\| B^q_{e,v_1}+A^q_e B^q_{e,v_2} \|_2+\|C^q_2 B^{q, \star}_{e,v_2}+T^q_2 \|_2+ \|C^q_2 (B^{q, \star}_{e,v_1}+ \overline{A}^q B^q_{e,v_2}) \|_2 ]$ and $\theta^q \triangleq \|A^q_e \|_2$, with $\overline{A}^q$, $A^q_e$, $B^{q}_{e,w}$, $B^{q, \star}_{e,w}$, $B^{q}_{e,v_1}$, $B^{q, \star}_{e,v_1}$, $B^{q}_{e,v_2}$ and $B^{q, \star}_{e,v_2}$ given in Lemma \ref{lem:resdef}.} 
% \end{array}%\\ \| C^q_2 \overline{A}^q {A^q_e}^{i}(B^q_{e,v_1}+ {A^q_e}B^q_{e,v_2}) \|_2
% \\ \nonumber \overline{\delta}^{q,tri}_{r} &\triangleq \eta_w [\| C^q_2 B^{q,\star}_{e,w} \|_2 +\| C^q_2 \overline{A}^q A^q_e\|_2 (1/(1-\theta^q))+ \beta^q]
%\\ \nonumber &+ \eta_v [\|C^q_2 B^{q, \star}_{e,v_2}+T^q_2 \|_2+ \|C^q_2 (B^{q, \star}_{e,v_1}+ \overline{A}^q B^q_{e,v_2}) \|_2 ],
%\\ \nonumber \beta^q &\triangleq ( \eta_w \| B^q_{e,w}\|_2+ \eta_v \| B^q_{e,v_1}+A^q_e B^q_{e,v_2} \|_2), \theta^q \triangleq \|A^q_e \|_2.
% \end{align}
% \end{small}
 \end{lem}
 \vspace{-0.3cm}
\begin{lem} \label{lem:mode_disjointness}
 Suppose that Assumption \ref{assumption:boundedness} holds. % there exists a known $R_y \in \mathbb{R}$ such that $\forall k, y_{k} \in Y \triangleq \{ y \in \mathbb{R}^{l} \vline \ \| y \|_2 \leq R_y \}$, i.e., there exists a known bound for the whole observation/measurement space \footnote{JUSTIFICATION...}.
  Consider two different modes $q \neq q' \in Q$ and their corresponding upper bounds for their residuals' norms, $\delta^{q}_{r,k }$ and $\delta^{q'}_{r,k }$, at time step $k$. At least one of the two modes $q \neq q'$ will be eliminated if
  
\vspace{-0.4cm}\begin{small}
 \begin{align} \label{eq:disjointness}
 %\begin{small}
  &\| C^q_2 \hat{x}^{\star,q}_{k|k}- C^{q'}_2 \hat{x}^{\star,q'}_{k|k}\hspace{-0.1cm}+\hspace{-0.1cm}D^{q}_{2} u^{q}_k-D^{q'}_{2} u^{q'}_k\|_2 \hspace{-0.1cm}>\hspace{-0.1cm}\delta^{q}_{r,k}\hspace{-0.1cm}+\hspace{-0.1cm}\delta^{q'}_{r,k }\hspace{-0.1cm}+\hspace{-0.1cm}R^{q,q'}_z
%  \\ \nonumber  &where \ R^{q,q'}_z \triangleq R_y\| T^q_2-T^{q'}_2\|_2.  
  %\end{small}
  \end{align}
  \end{small}\vspace{-0.4cm}
  
\noindent  where $R^{q,q'}_z \triangleq R_y\| T^q_2-T^{q'}_2\|_2$. 
\end{lem} 
 \vspace{-0.3cm}
\begin{lem} \label{lem:resdef2}
Consider any mode $q$ with the unknown true mode being $q^{*}$. Then, at time step $k$, we have
%Consider mode $q$ at time step $k$, and the unknown true mode $q^{*}$. Then, %is the true mode, then the residual signal at time step $k$ can be obtained as
\begin{align*}
% r^{q|*}_k &= C^q_2 \tilde{x}^{\star,q}_{k|k}+v^q_{2,k}=\mathbb{A}^q_k {t}_k, \label{eq:resid_ideal}\\ 
r^{q}_k &= \begin{bmatrix} \mathbb{T}^{q,q^*}_k & \mathbb{B}^{q,q^*}_k & \mathbb{D}^{q,q^*}_k  \end{bmatrix} \begin{bmatrix} t^{\top}_k & u^{q^*\top}_{0:k} & d^{q*\top}_{0:k} \end{bmatrix}^{\top}, %\label{eq:resid_comp}
 \end{align*}
 where $u^{q^*}_{0:k} \triangleq \begin{bmatrix} u^{q*\top}_k & u^{q*\top}_{k-1} & \dots u^{q*\top}_0  \end{bmatrix}^{\top}$,
 \vspace{-0.15cm}
 \begin{small}
 \begin{align*}
% \nonumber &{t}_k \triangleq \begin{bmatrix} \tilde{x}^{\top}_{0|0} & w^{\top}_0 & \dots & w^{\top}_{k-1} & v^{\top}_0 & \dots & v^{\top}_k \end{bmatrix}^{\top} \in \mathbb{R}^{n+kn+(k+1)l},\\
  %\nonumber u^{q^*}_{0:k} &\triangleq \begin{bmatrix} u^{q*\top}_k & u^{q*\top}_{k-1} & \dots u^{q*\top}_0  \end{bmatrix}^{\top}\hspace*{-0.15cm},\\ % d^{q*}_{0:k} \triangleq \begin{bmatrix} d^{q*\top}_k & d^{q*\top}_{k-1} & \dots d^{q*\top}_0  \end{bmatrix}^{\top}\hspace{-0.15cm},\\
% \end{align}
%
% \begin{align}
  \nonumber \ \mathbb{T}^{q,q^*}_k &\triangleq (T^{q^*}_2-T^q_2)\begin{bmatrix} CA^k & CA^{k-1} & \dots & C & I \end{bmatrix} +\mathbb{A}^q_k,
\\ \nonumber \ \mathbb{B}^{q,q^*}_k &\triangleq (T^{q^*}_2-T^q_2)\begin{bmatrix} D & CB & CAB & \dots & CA^{k-1}B \end{bmatrix},% +\mathbb{A}^q_k,
 \\ \nonumber \ \mathbb{D}^{q,q^*}_k &\triangleq (T^{q^*}_2-T^q_2)\begin{bmatrix} H & CG & CAG \dots & CA^{k-1}G \end{bmatrix},% +\mathbb{A}^q_k,
% \\ \nonumber u^{q^*}_{0:k} &\triangleq \begin{bmatrix} u^{q*\top}_k & u^{q*\top}_{k-1} & \dots u^{q*\top}_0  \end{bmatrix}^{\top}, d_{0:k} \triangleq \begin{bmatrix} d^{\top}_k & d^{\top}_{k-1} & \dots d^{\top}_0  \end{bmatrix}^{\top},
  \end{align*}
   \end{small}\vspace{-0.4cm}
   % \\ \nonumber &t_k \triangleq [\tilde{x}_{0|0} \ w_0 \ ... \ w_{k-1} \ v_0 \ ... \ v_k] \top.
%$B^\star_{e,w}\triangleq I-G_2 M_2 C_2 $, $B^\star_{e,v1}\triangleq -(I-G_2 M_2 C_2) (G_1 M_1 T_1)$ %and $B^\star_{e,v2}\triangleq -G_2 M_2T_2$.

\noindent  with $t_k$ given in Lemma \ref{lem:resdef} and $d^{q*}_{0:k} $ in Assumption \ref{as:2}.
\end{lem}
%Now, we are ready to state some sufficient conditions for mode detectability.
\vspace{-0.1cm}
\begin{thm}[Sufficient Conditions for Mode Detectability] \label{thm:strong_mode_detect}
%Suppose the true mode attack signal has unlimited energy, i.e., $\lim_{k\to\infty} \|d^{q*}_{0:k}\|_2 = \infty$. Then, 
System \eqref{eq:sys_desc} is mode detectable, i.e., all false modes will be eliminated after some large enough time step $K$, using Algorithm \ref{algorithm1}, %applying \color{red}the mode elimination approach\color{black} in Section \ref{sec:MainResult}, 
if the assumptions in Theorem \ref{thm:filterbanks} and either of the following hold: 
\renewcommand{\theenumi}{\roman{enumi}}
\begin{enumerate}
\item Assumption \ref{assumption:boundedness} and $ \forall q,q' \in Q$, $q\neq q'$, \label{item:second}% \ s.t. \ 
\begin{align}
\nonumber  \sigma_{min} (W^{q,q'}) > \frac{\overline{\delta}^{q,tri}_{r}+\overline{\delta}^{q',tri}_{r}+R^{'q,q'}_y}{\sqrt{R^2_x+\eta^2_v}};\, % $\forall z \in \mathbb{C}, |z| 
\end{align}
\item \label{item:first} Assumption \ref{as:2} and \yong{$T^q_2 \neq T^{q'}_2$ holds $ \forall q,q' \in Q, q \neq q'$},  % \moha{$\implies$} 
 %\moha{.}%\implies T^q_2 \neq T^{q'}_2$. 
\end{enumerate} 
where $W^{q,q'}\hspace{-0.1cm} \triangleq \hspace{-0.1cm}\begin{bmatrix} (C^q_2 - C^{q'}_2) & (T^q_2 - T^{q'}_2) & -I  & I & D^q_2  & -D^{q'}_2 \end{bmatrix}$.
\end{thm}
\vspace{-0.2cm}
\section{Simulation Example} \label{sec:examples}
\vspace{-0.1cm}
We consider a system  that has been used as a benchmark for many state and input filters/observers (e.g.,\cite{yong2018switching}):%, which is similar to the failure detection problem first considered in \cite{Keller.1996}, given by:

\small \vspace{-0.2cm}
\begin{align*}
A &=\hspace{-0.05cm} \begin{bmatrix} 0.5 & 2 & 0 & 0 & 0\\ 0 & 0.2 & 1 & 0 &1 \\ 0 & 0 & 0.3 & 0 & 1 \\ 0 & 0 & 0 & 0.7 & 1 \\ 0 & 0 & 0 & 0 & 0.1\end{bmatrix}\hspace{-0.1cm}; 
 G =\hspace{-0.05cm} \begin{bmatrix} 1  \\ 0.1  \\ 0.1 \\1\\0 \end{bmatrix}\hspace{-0.1cm}; H=\hspace{-0.05cm}\begin{bmatrix} 1 & 0 & 0 &0 \\ 0 & 1 & 0 & 0 \\ 0 & 0 & 1 & 0 \\ 0 & 0 & 0 & 1 \\ 0&0&0&0 \end{bmatrix}\hspace{-0.1cm};\\ 
 B &= 0_{5 \times 1}; C = I_5;D = 0_{5 \times 1}.
\end{align*} %\vspace{-0.05cm}
\normalsize 

The unknown inputs used in this example are as given in Figure \ref{fig:estimates}, while the initial state estimate and noise signals have bounds $\delta_x=0.5$, $\eta_w=0.02$ and $\eta_v= 10^{-4}$. We assume possible attacks on the actuator and four of five sensors, i.e., $t_a=1$ and $t_s=4$. Moreover, we assume that there are $\rho=4$ attacks, so we should consider $Q={5 \choose 4}\yong{=5}$ modes. Table \ref{fig:mode_table} indicates different modes, their attack location(s) and the matrix $T^q_2$ for each mode $q$, where, as can be observed, the \yong{second} %first 
set of sufficient conditions in Theorem \ref{thm:strong_mode_detect} holds, i.e., \yong{$T^q_2 \neq T^{q'}_2$ for all $q \neq q'$,} %$q \neq q' \implies T^q_2 \neq T^{q'}_2$, 
so we expect that after some large enough time, all the false modes be eliminated, i.e., at most one (true) mode remains at each time step, which can be seen in Figure \ref{fig:resbounds}, where the number of eliminated modes at each time step is exhibited.
%.\begin{figure}[!h]
%\begin{center}
%\includegraphics[scale=0.450,trim=0mm 0mm 0mm 0mm,clip]{Figures/mode_table.pdf}%\vspace{-0.15cm}
%\caption{Different modes and their $T^q_2$ \label{fig:mode_table}} 
%\includegraphics[scale=0.35,trim=17mm 60mm 20mm 2mm,clip]{Figures/bounds_ULISO.eps}%\vspace{-0.15cm}
%\caption{Actual estimation errors and radii of set-valued estimates of states, \hspace{-0.1cm}$\|\tilde{x}_{k|k}\|$, \hspace{-0.075cm}$\delta^x_k$, and unknown inputs, \hspace{-0.1cm}$\|\tilde{d}_{k}\|$, \hspace{-0.075cm}$\delta^d_k$. \hspace{-0.45cm} \label{fig:variances}}
%\end{center}
%\vspace{-0.425cm}
%\end{figure}
\begin{table}[b] \vspace{-0.05cm}
	\centering \setlength{\tabcolsep}{3pt}
	\caption{Different modes and their $T^q_2$. \label{fig:mode_table}} %\vspace{0.1cm}
	\vspace{0.05cm} \footnotesize \fontsize{7.5}{7.5}\selectfont
	\begin{tabular}{| c | c | c |}
		\hline
		  Mode & Attack location(s) & $T_2^q$ \\ \hline
		$q=1$ & Actuator \& Sensors 1,2,3 & [0.2518  -0.1068  -0.2409  -0.5862  0.7236]$^\top$ \\ \hline
		$q=2$ & Actuator \& Sensors 1,2,4 & [0.0080  0.7604  -0.1522  -0.5862  -0.6313]$^\top$ \\ \hline
		$q=3$ & Actuator \& Sensors 1,3,4 & [-0.5357  0.7289  0.1984  -0.3774  0.0009]$^\top$ \\ \hline
		$q=4$ & Actuator \& Sensors 2,3,4 & [0.7092  -0.5570  -0.1797  -0.3295  0.2143]$^\top$ \\ \hline
		$q=5$ & Sensors 1,2,3,4 & [0.1679  -0.5682  0.5198  -0.4883  0.3747]$^\top$ \\ \hline
	\end{tabular} \normalsize \vspace{-0.1cm}
\end{table}
\begin{figure}[!b]
\begin{center}
%\includegraphics[scale=0.495,trim=0mm 0mm 0mm 0mm,clip]{Figures/mode_table2.pdf}%\vspace{-0.15cm}
%\caption{Different modes and their $T^q_2$ \label{fig:mode_table}}
\includegraphics[scale=0.205,trim=43mm 8mm 20mm 5mm,clip]{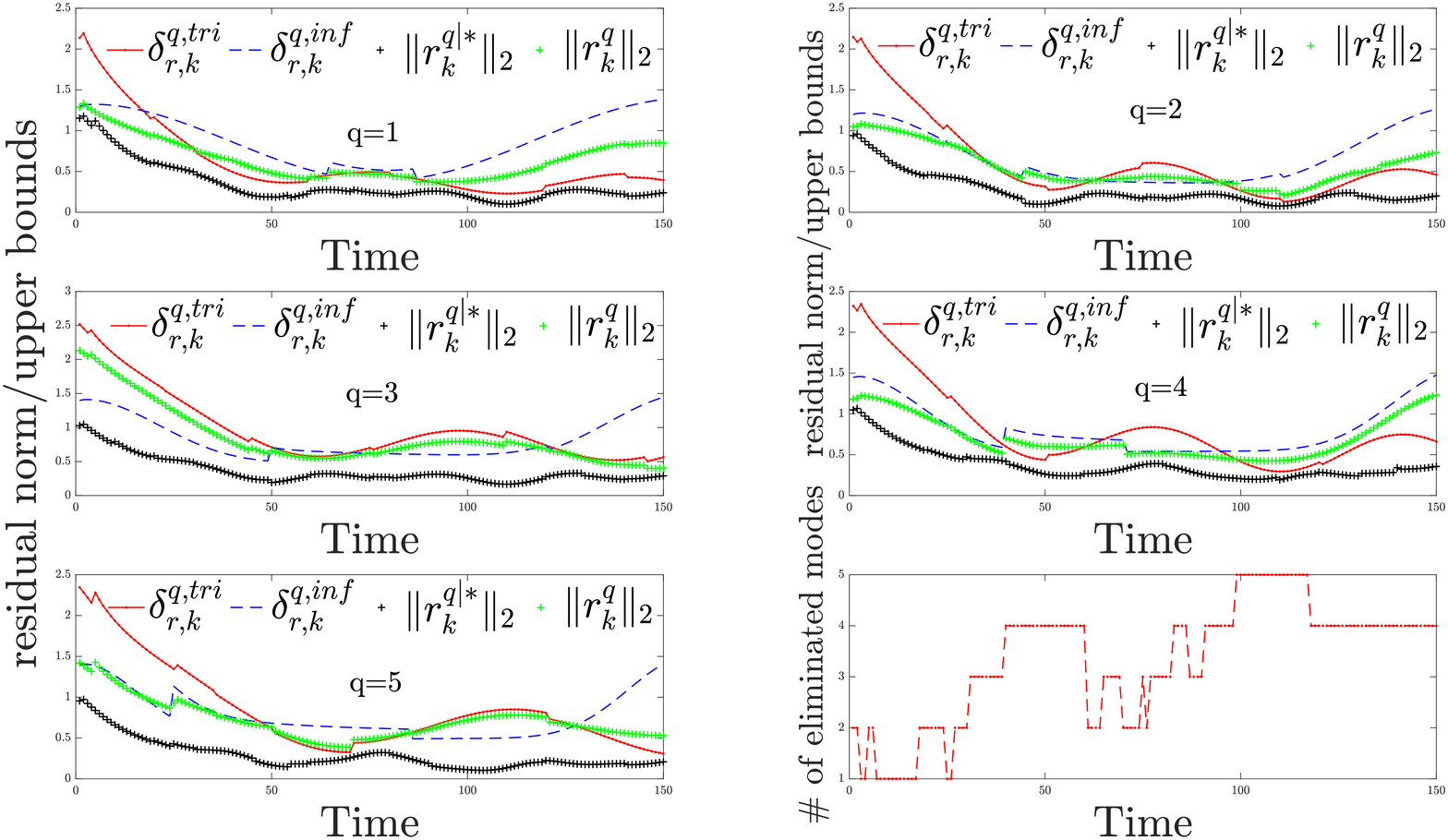}%\vspace{-0.15cm}
\caption{$\|r^{q}_{r,k}\|_2$,$\|r^{q|*}_{r,k}\|_2$ and \moh{their} upper bounds for different modes, as well as the number of eliminated modes in time \label{fig:resbounds} }
%\includegraphics[scale=0.200,trim=36mm 18mm 20mm 17mm,clip]{Figures/estimates.eps}%\vspace{-0.15cm}
%\caption{$x_k$ and $\hat{x}^q_{k|k}$ for $q=3$ \label{fig:estimates} }
%\vspace{.5cm}
\includegraphics[scale=0.205,trim=45mm 0mm 20mm 5mm,clip]{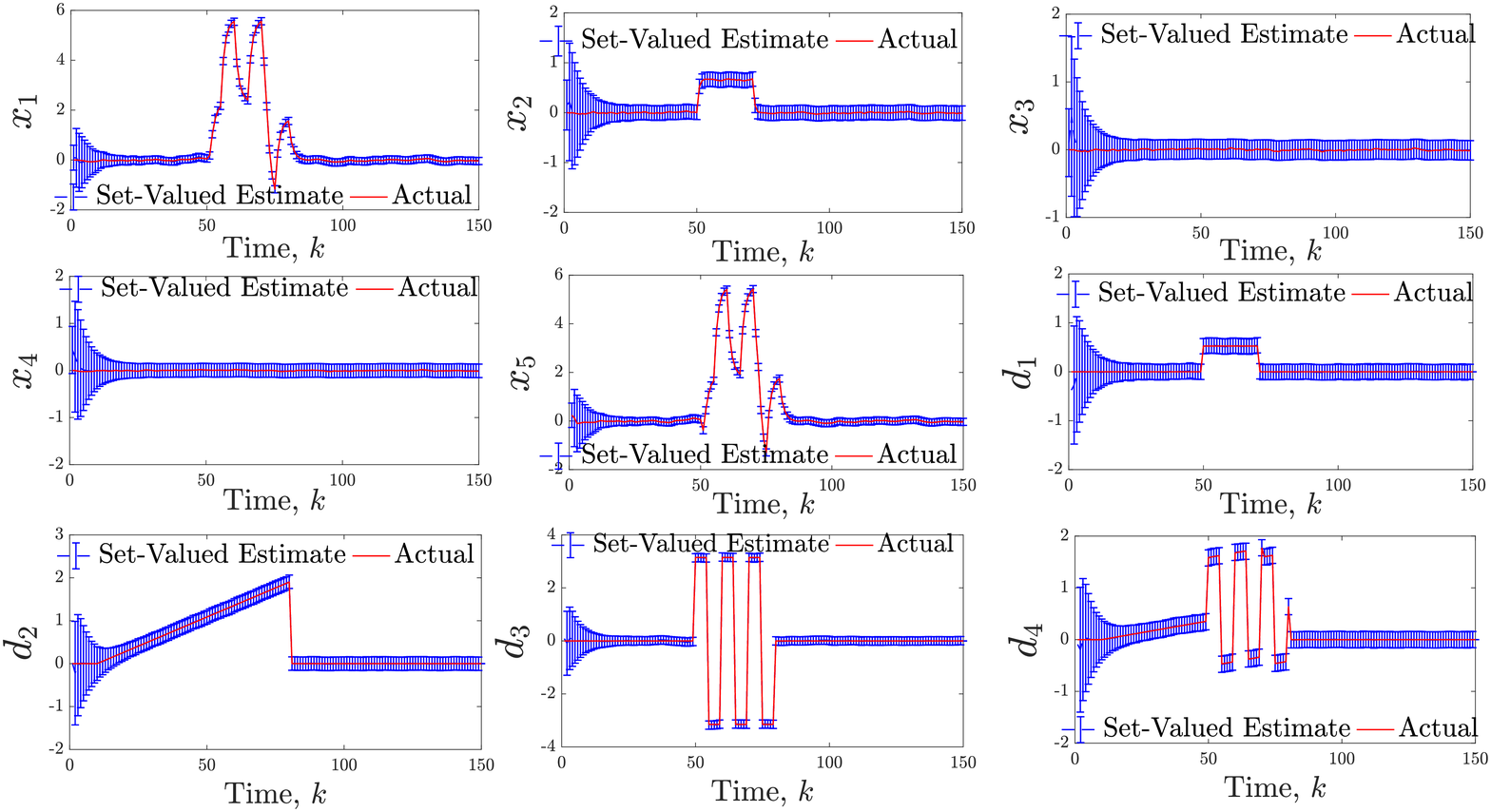}%\vspace{-0.15cm}
%\vspace{.5cm}
\caption{State and unknown input set-valued estimates. \label{fig:estimates} \vspace{-0.3cm} }
%\includegraphics[scale=0.35,trim=17mm 60mm 20mm 2mm,clip]{Figures/bounds_ULISO.eps}%\vspace{-0.15cm}
%\caption{Actual estimation errors and radii of set-valued estimates of states, \hspace{-0.1cm}$\|\tilde{x}_{k|k}\|$, \hspace{-0.075cm}$\delta^x_k$, and unknown inputs, \hspace{-0.1cm}$\|\tilde{d}_{k}\|$, \hspace{-0.075cm}$\delta^d_k$. \hspace{-0.45cm} \label{fig:variances}}
\end{center}
%\vspace{-0.425cm}
\end{figure}
Moreover, for each specific mode $q$, the signals $\|r^q_{k}\|_2,\|r^{q|*}_{k}\|_2,\delta^{q,tri}_{r,k}$ and $\delta^{q,inf}_{r,k}$ are depicted in  Figure \ref{fig:resbounds}. % and it's upper bounds, for a specific mode $q_0$, when $\rho=2$ and $\mathbb{I}_{q_0}=\begin{bmatrix} 1 & 0 & 0 \\ 0 & 0 & 0 \\ 0 & 0 & 1 \end{bmatrix}^{\top}$. 
 As can be seen, up to some large enough time, at different time intervals for different modes, one of the upper bounds may be tighter than the other, or vice-versa, so it is reasonable that we consider a minimum of them as the computed upper bound in our mode elimination algorithm. Furthermore, for all modes, $\delta^{q,tri}_{r,k}$ is eventually  convergent while $\delta^{q,inf}_{r,k}$ diverges, as we proved in Lemma \ref{lem:delta_inf_diverge}. So, after some large enough time, $\delta^{q,tri}_{r,k}$ can be used as our upper-bound, while $\delta^{q,inf}_{r,k}$ becomes useless. The %simulation results for the 
corresponding set-valued estimates are provided in Figure \ref{fig:estimates}.% given that the true mode is $q=3$, i.e., there are attacks on the actuator and sensors 1,3 and 4. 
  %< \delta^{q_0,tri}_{r,k}$, so it may be more reasonable to use $\delta^{q_0,inf}_{r,k}$ rather than $\delta^{q_0,tri}_{r,k}$, but as time increases, $\delta^{q_0,inf}_{r,k}$ gradually diverges while $\delta^{q_0,tri}_{r,k}$ converges to some steady state, so asymptotically speaking, $\delta^{q_0,inf}_{r,k}$ will become useless and $\delta^{q_0,tri}_{r,k}$ should be used.  

%\begin{figure}[!h]
%\begin{center}
%\includegraphics[scale=0.15,trim=6mm 0mm 0mm 0mm,clip]{Figures/mode_elimination_1.eps}%\vspace{-0.15cm}
%\caption{Number of eliminated modes for $\rho=1$ \label{fig:modeliminate1}}
%\includegraphics[scale=0.15,trim=6mm 0mm 0mm 0mm,clip]{Figures/mode_elimination_2.eps}%\vspace{-0.15cm}
%\caption{Number of eliminated modes for $\rho=2$ \label{fig:modeliminate2}}
%\end{center}
%\vspace{-0.425cm}
%\end{figure}
 %Figures \ref{fig:modeliminate1} and \ref{fig:modeliminate2}, exhibit the number of eliminated modes, using both upper bounds for $\rho=1$ and $\rho=2$ respectively.
\vspace{-0.1cm}

\section{Conclusion} \label{sec:conclusion}

We proposed a residual-based approach for %bounded-error 
hidden mode switched linear systems with \yong{bounded-norm noise and} %completely 
unknown attack signals\moh{. The proposed approach} at each time step, removes the inconsistent modes and their corresponding observers from a bank of estimators, which includes mode-matched observers\moh{. Each mode-matched observer}, conditioned on its corresponding mode being true, simultaneously finds bounded sets of states and unknown inputs that include the true state and inputs. Our mode elimination criterion required a bounded upper bound for the residual's norm, for which we proved its existence and computed it by over-approximating the value function of a non-concave NP-hard norm-maximization problem by expanding its constraint set and converting \moh{it} into a convex maximization over a convex set with finite number of extreme points. Such a problem can be solved by enumerating the objective function on the extreme points of the constraint set and comparing the corresponding values. Moreover, we proved the convergence of the upper bound signal %. Furthermore, applying matrix lower bound theorem and considering the convergence behavior of the provided  upper bound signal, we 
		and derived sufficient conditions for eventually eliminating all false modes \yong{using} our \moh{mode elimination algorithm}. 
		Finally, we demonstrated  the effectiveness of our \moh{observer} using an illustrative example.
\bibliographystyle{unsrturl}
\vspace{-0.1cm}
\bibliography{biblio}
\normalsize

\vspace{-0.4cm}\section*{Appendix: Proofs} \label{subsec:thmproof}

\vspace{-0.2cm}  \begin{proof}[Proof of Lemma \ref{lem:delta_inf_diverge}]
 To show \eqref{eq:r_inf_diverge}, we first find a lower bound for $\delta^{q,inf}_{r,k}$. Then, we show that the lower bound diverges and so does $\delta^{q,inf}_{r,k}$. Define $\tilde{t}^{\star}_k \triangleq {t^{\star}_k}/{\eta^t_k}$, where $\eta^t_k$ is defined in Corollary \ref{cor:verticenorm}. Now consider  %This Corollary also implies that $ \| \tilde{t}^{\star}_k \|_2=1$. Now consider
% \begin{scriptsize}
 \begin{align*}
 \eta^t_k \sigma_{min}(\mathbb{A}^q_k)& =  \sigma_{min}(\eta^t_k \mathbb{A}^q_k)= \min \limits_{\| t \|_2 \leq 1} \| \eta^t_k \mathbb{A}^q_k t \|_2 
\\ & \leq \| \eta^t_k \mathbb{A}^q_k \tilde{t}^{\star}_k\|_2 = \| \mathbb{A}^q_k t^{\star}_k \|_2 =\delta^{q,inf}_{r,k},
 \end{align*}
 %\end{scriptsize}
 where $\sigma_{min}(A)$ is the least non-trivial singular value of matrix $A$, the first equality holds since $\sigma_{min}(.)$ is a linear operator, the second equality is a special case of a \emph{matrix lower bound} \cite{grcar2010matrix} when 2-norms are considered, the inequality holds since $\|\tilde{t}^{\star}_k\|_2=1$ by Corollary \ref{cor:verticenorm}, so $\tilde{t}^{\star}_k$ is a feasible point for the minimization in the third statement and the last equality holds by Theorem \ref{thm:resid_comp_up_bound}. So far we have shown that $\eta^t_k \sigma_{min}(\mathbb{A}^q_k)$ is a lower bound for $\delta^{q,inf}_{r,k}$. \moh{Next,} we will \moh{prove} that $\eta^t_k \sigma_{min}(\mathbb{A}^q_k)$ is unbounded. First, it is trivial that $\eta^t_k$ is unbounded by its definition \moh{in} Corollary \ref{cor:verticenorm}. Second, consider the block matrix $\mathbb{A}^q_k$ in Lemma \ref{lem:resdef}. \moh{By} \yong{the} strong detectability assumption, matrix $A^q_e$ is stable \cite[Theorem 3 and Appendix C]{yong2018simultaneous}, so all the block matrices of  $\mathbb{A}^q_k$, except \moh{three} of them which are constant matrices with respect to time, converge to zero matrices when time goes to infinity. Hence $\mathbb{A}^q_k$ converges to an infinite dimensional sparse matrix, with only \moh{three} non-zero finite dimensional constant blocks and so the limit matrix has a finite rank and clearly has a bounded minimum non-trivial singular value. Henceforth, $\eta^t_k \sigma_{min}(\mathbb{A}^q_k)$ is unbounded, since the product of \moh{the} bounded and non-zero $\sigma_{min}(\mathbb{A}^q_k)$ and unbounded $\eta^t_k$ is unbounded. 
 As for \eqref{eq:r_tri_converge}, the first equality holds by definition of $\hat{\delta}^{q}_{r,k}$ (\yong{cf.} Theorem \ref{thm:resid_comp_up_bound}) and \eqref{eq:r_inf_diverge}, the first inequality holds \moh{since} $\delta^{q,tri}_{r,k}\leq \overline{\delta}^{q,r}_{r,k}$ by triangle and sub-multiplicative inequalities and the last equality, i.e., convergence of $\delta^{q,tri}_{r,k}$, follows from strong detectability assumption which implies the stability of ${A^q_e}$ \cite[Theorem 3]{yong2018simultaneous}.% and the use of sub-multiplicative and triangle inequalities. The equality also holds by boundedness of $\delta^{r,q}_{k,tri}$ and divergence of $\delta^{r,q}_{k,inf}$.
 \end{proof}
\vspace{-0.5cm} \begin{proof}[Proof of Lemma \ref{lem:mode_disjointness}] 
 Suppose, for contradiction, that none of $q$ and $q'$ are eliminated. Then 
 \begin{small}
 \begin{align}
 \nonumber &\| C^q_2 \hat{x}^{\star,q}_{k|k}+D^{q}_{2} u^{q}_k- C^{q'}_2 \hat{x}^{\star,q'}_{k|k}-D^{q'}_{2} u^{q'}_k\|_2=
 \\ \nonumber &\|r^{q'}_{k}-r^{q}_{k}+z^q_{2,k}-z^{q'}_{2,k})\|_2 \leq \|r^{q'}_{k}\|_2+\|r^{q}_{k}\|_2+\|z^q_{2,k}-z^{q'}_{2,k}\|_2 
 \\ \nonumber &\leq \delta^{q}_{r,k}+\delta^{q'}_{r,k } + R_y\| T^q_2-T^{q'}_2\|_2, 
 %\\ \nonumber &\leq \delta^{q}_{r,k}+\delta^{q'}_{r,k } + 2R_z,
 \end{align}
 \end{small}
 
\vspace{-0.4cm}\noindent where the equality holds by %adding and subtracting $z^q_{2,k}$ and $z^{q'}_{2,k}$ to the first term and 
Definition \ref{defn:computedresidual}, the first inequality holds by triangle inequality and the last inequality holds by the assumption that none of $q$ and $q'$ can be eliminated, as well as the boundedness assumption for the measurement space. This %completes the proof since the 
last inequality contradicts with the inequality in the lemma, thus the result holds. %\eqref{eq:disjointness}.    
\end{proof}
\vspace{-0.35cm}\begin{proof}[Proof of Lemma \ref{lem:resdef2}]
%Considering \eqref{eq:resid_ideal}, the first equality directly comes from Definition \ref{defn:computedresidual} and eqaution \eqref{eq:z2}, assuming that $q$ is the true mode, and the second equality is implied by the first equality and the fact that %(c.f., \cite[Appendix C]{yong2018simultaneous}) %&\begin{array}{l}
%\begin{align}
%\nonumber \tilde{x}^{\star,q}_{k|k}&=\overline{A}^q{A^q_e}^{k-1} \tilde{x}_{0|0}+\overline{A}^q{A^q_e}^{k-2}\begin{bmatrix} B^q_{e,w}  B^q_{e,v1} \end{bmatrix}\vec{w}_0 %\begin{bmatrix} w_0\\ v_0 \end{bmatrix}\\
% \\ \nonumber &+B_{e,w}^{\star,q} w_{k-1} + (B_{e,v1}^{\star,q}+\overline{A}^qB^q_{e,v2}) v_{k-1} + B_{e,v2}^{\star,q} v_k
%  \\ \nonumber &+\textstyle \sum_{i=1}^{k-2} \overline{A}^q{A^q_e}^{k-1-i} \begin{bmatrix} B^q_{e,w} & B^q_{e,v1}+A^q_e B^q_{e,v2} \end{bmatrix} \vec{w}_i, 
%   \\ \nonumber \vec{w}_k&\triangleq \begin{bmatrix} w_k^\top & v_k^\top\end{bmatrix}^\top. %\begin{bmatrix} w_i\\ v_i \end{bmatrix},
%  \end{align} (c.f., \cite[Appendix C]{yong2018simultaneous}).
The result can be obtained by applying Proposition \ref{prop:residecomposition},  \eqref{eq:resid_ideal} and the closed-form output signal:
  
  \begin{scriptsize}\vspace{-0.3cm}
  \begin{align}
  \nonumber y_k&=\begin{bmatrix} \begin{bmatrix} (CA^k)^{\top} \\ (CA^{k-1})^{\top} \\ \vdots \\ C^{\top} \\ I \end{bmatrix}^{\top} & \begin{bmatrix} H^{\top} \\ (CG)^{\top} \\ (CAG)^{\top} \\ \vdots \\ (CA^{k-1}G)^{\top} \end{bmatrix}^{\top} 
  & \begin{bmatrix} D^{\top} \\ (CB)^{\top} \\ (CAB)^{\top} \\ \vdots \\ (CA^{k-1}B)^{\top} \end{bmatrix}^{\top}  \end{bmatrix} \begin{bmatrix} t_k \\[-0.2cm] \\ d^{q*}_{0:k} \\[-0.2cm] \\ u^{q*}_{0:k} \end{bmatrix},%\begin{bmatrix}
  \end{align}
  \end{scriptsize}\vspace{-0.2cm}
  
\noindent  which can be derived by using \eqref{eq:sys_desc} and simple induction. %This, completes the proof.
 \end{proof}
 
\vspace{-0.35cm} \begin{proof}[Proof of Theorem \ref{thm:strong_mode_detect}]
To show that \eqref{item:second} is sufficient for asymptotic mode detectability, consider Lemma \ref{lem:mode_disjointness} with $\delta^{q,tri}_{r,k}$ as the upper bound. It suffices to show $\exists K \in \mathbb{N}$, such that \eqref{eq:disjointness} holds for $k \geq K, \forall q \neq q' \in \mathbb{Q}.$ 
 Notice that by Definition \ref{defn:computedresidual}, $C^q_2 \hat{x}^{\star,q}_{k|k}=C^q_2 x_k +T^q_2 v_k - r^{q|*}_k$.
% \begin{align}
 % \nonumber  r^{q|*}_k &= C^q_2 \tilde{x}^{\star,q}_{k|k}+v^q_{2,k}= C^q_2 x_k -C^q_2 \hat{x}^{\star,q}_{k|k}+T^q_2 v_k
  % \\ \nonumber &\implies C^q_2 \hat{x}^{\star,q}_{k|k}=C^q_2 x_k +T^q_2 v_k - r^{q|*}_k.
 % \end{align}
   Plugging this into \eqref{eq:disjointness}, we need to show $\exists K \in \mathbb{N}$ such that: 

\begin{small}\vspace{-0.4cm}
\begin{align} 
 &\| W^{q,q'} s^{q,q'}_k \|_2 >\delta^{q,tri}_{r,k}+\delta^{q',tri}_{r,k}+R^{q,q'}_z, \forall k \geq K, \label{eq:sufficient}% q \neq q' \in Q, %x_k \in X,
%\\ \nonumber &\| (C^q_2 - C^{q'}_2) x_{k}+(T^q_2-T^{q'}_2)v_k-(r^q_k-r^{q'}_k)+ D^{q}_{2} u^{q}_k-D^{q'}_{2} u^{q'}_k\|_2 
\\ \nonumber &s^{q,q'}_k \triangleq \begin{bmatrix} x^{\top}_k & v^{\top}_k & r^{q|* \top}_k & r^{q'|* \top}_k & u^{q \top}_k & u^{q' \top}_k \end{bmatrix}^{\top},\forall q \neq q' \in \mathbb{Q}.   
\end{align} 
\end{small}
A sufficient condition to satisfy \eqref{eq:sufficient} is that $\exists K \in \mathbb{N}$ such that $\forall k \geq K$, \eqref{eq:sufficient} holds for all $s^{q,q'}_k$. Equivalently, it suffices %that:
 \begin{align}
 \nonumber &\min \limits_{ x_k,v_k,r^q_k,r^{q'}_k } \| W^{q,q'} s^{q,q'}_k \|_2>\delta^{q,tri}_{r,k}+\delta^{q',tri}_{r,k}+R^{q,q'}_z
 \\ \nonumber & s.t. \ \|x_k\|_2 \leq R_x, \|v_k\|_2 \leq \eta_v, \|r^{q|*}_k\|_2 \leq \delta^{q,tri}_{r,k}, 
 \\ \nonumber &\ \ \ \ \ \|r^{q'|*}_k\|_2 \leq \delta^{q',tri}_{r,k}, \ \forall k \geq K, \forall q \neq q' \in \mathbb{Q}.
 \end{align}
 By expanding the constraint set, it is sufficient to require that $\exists K \in \mathbb{N}$ such that:
  \begin{align}
 \nonumber &\min \limits_{s^{q,q'}_k} \| W^{q,q'} s^{q,q'}_k \|_2>\delta^{q,tri}_{r,k}+\delta^{q',tri}_{r,k}+R^{q,q'}_z %\ s.t. 
 \\ \nonumber s.t. \ &\| s^{q,q'}_k\|^2_2 \leq R^2_x \hspace{-0.1cm}+\hspace{-0.1cm} \eta^2_v \hspace{-0.1cm}+\hspace{-0.1cm} (\delta^{q,tri}_{r,k})^2 \hspace{-0.1cm} +\hspace{-0.1cm} (\delta^{q',tri}_{r,k})^2+(u^q_k)^2+(u^{q'}_k)^2 
 \\ \nonumber &\forall k \geq K,\forall q \neq q' \in \mathbb{Q}.
 \end{align}
 Now, by \emph{matrix lower bound} theorem\moh{\cite{grcar2010matrix}} and similar argument as in the proof of Lemma \ref{lem:delta_inf_diverge}, it is sufficient to be satisfied that $ \exists K \in \mathbb{N} \ s.t. \ \forall k \geq K, \forall q \neq q' \in \mathbb{Q} : $
 
\vspace{-0.3cm}
 \begin{footnotesize}
 \begin{align} \label{eq:timed_sufficient_compatibility} 
 \sigma^2_{min} (W^{q,q'}) \hspace{-0.1cm}>\hspace{-0.1cm} \frac{ (\delta^{q,tri}_{r,k}+ \delta^{q',tri}_{r,k}+R^{q,q'}_z)^2}{R^2_x\hspace{-0.1cm}+\hspace{-0.1cm}\eta^2_v\hspace{-0.1cm}+\hspace{-0.1cm}( \delta^{q,tri}_{r,k})^2\hspace{-0.1cm}+\hspace{-0.1cm}( \delta^{q',tri}_{r,k})^2\hspace{-0.1cm}+\hspace{-0.1cm}(u^q_k)^2\hspace{-0.1cm}+\hspace{-0.1cm}(u^{q'}_k)^2}.
 \end{align}
 \end{footnotesize}
 
\vspace{-0.1cm}  \eqref{eq:timed_sufficient_compatibility} provides us a \emph{time-dependent} sufficient condition for mode detectability. In order to find a \emph{time-independent} sufficient condition, notice that $ \frac{ (\overline{\delta}^{q,tri}_{r,k}+ \overline{\delta}^{q',tri}_{r,k}+R^{q,q'}_z)^2}{R^2_x+\eta^2_v}$ is an upper bound for the right hand side of \eqref{eq:timed_sufficient_compatibility}, since the latter's denominator is smaller than the former's and the numerator of the latter is an upper bound signal for the former's  
  %consider that 
  by \moh{triangle and sub-multiplicative inequalities}. So a sufficient condition for \eqref{eq:timed_sufficient_compatibility} is $ \exists K \in \mathbb{N} \ s.t. \ \forall k \geq K, \forall q \neq q' \in \mathbb{Q} : $
    %\begin{small}
 \vspace{-0.1cm} \begin{align} \label{eq:timed_sufficient_compatibility_2}
 \sigma^2_{min} (W^{q,q'}) > \frac{ (\overline{\delta}^{q,tri}_{r,k}+ \overline{\delta}^{q',tri}_{r,k}+R^{q,q'}_z)^2}{R^2_x+\eta^2_v}.
 \end{align}
Then, for the above to hold, it suffices that 
 $$\sigma^2_{min} (W^{q,q'}) > \lim_{k\to \infty} \frac{ (\overline{\delta}^{q,tri}_{r,k}+ \overline{\delta}^{q',tri}_{r,k}+R^{q,q'}_z)^2}{R^2_x+\eta^2_v},$$
 which is equivalent to \eqref{item:second} by \eqref{eq:r_tri_converge}.
As for the sufficiency of \eqref{item:first}, notice that by Theorems \ref{thm:online_mode_elimination} and \ref{thm:resid_comp_up_bound}, Lemma \ref{lem:resdef} and Definition \ref{defn:strong_mode_detecatble}, for mode detectability, it suffices that for any specific mode $q$, the true mode $q^*$ and large enough $k$,

\vspace{-0.4cm}
\begin{small}
\begin{align}
%\vspace{-2.5cm}
\nonumber \|r^q_k\|_2=\|\begin{bmatrix} \mathbb{T}^{q,q^*}_k & \mathbb{B}^{q,q^*}_k & \mathbb{D}^{q,q^*}_k  \end{bmatrix} \begin{bmatrix} t^{\top}_k & u^{q^*\top}_{0:k} & d^{q*\top}_{0:k} \end{bmatrix}^{\top}\|_2>\delta^{q,tri}_{r,k},
 \end{align}
 \end{small}  
\vspace{-0.4cm} 

\noindent\moh{with $t_k$ given in \eqref{eq:inf_norm}.} Since $q^*$ is unknown, a sufficient condition \yong{to satisfy} %\moh{to imply} 
the above equality is $ \forall q' \neq q \in Q :$
 \begin{align}
\nonumber \|r^q_k\|_2=\|\begin{bmatrix} \mathbb{T}^{q,q'}_k & \mathbb{B}^{q,q'}_k & \mathbb{D}^{q,q'}_k  \end{bmatrix} \begin{bmatrix} t^{\top}_k & u^{q'\top}_{0:k} & d^{q*\top}_{0:k} \end{bmatrix}^{\top}\|_2>\delta^{q,tri}_{r,k}.
 \end{align}
So it suffices that $ \forall q' \neq q \in Q, \exists \overline{d} \in \mathbb{R}$, such that: 
  \begin{align}
 \nonumber &\min \limits_{t'_k} \|\begin{bmatrix} \mathbb{T}^{q,q'}_k & \mathbb{B}^{q,q'}_k & \mathbb{D}^{q,q'}_k  \end{bmatrix} \moh{t'_k}\|_2>\delta^{q,tri}_{r,k}
 \\[-0.2cm] \nonumber &s.t. \ t'_k=\begin{bmatrix} t^{\top}_k & u^{q'\top}_{0:k} & d^{q*\top}_{0:k} \end{bmatrix}^{\top},\|d^{q*}_{0:k}\|_2 \geq \overline{d},
 %t_k= \begin{bmatrix} \tilde{x}^{\top}_{0|0} & w^{\top}_0 & \dots & w^{\top}_{k-1} & v^{\top}_0 & \dots & v^{\top}_k \end{bmatrix},
 \\[-0.05cm] \nonumber &\ \ \ \ \ t_k= \begin{bmatrix} \tilde{x}^{\top}_{0|0} & w^{\top}_0 & \dots & w^{\top}_{k-1} & v^{\top}_0 & \dots & v^{\top}_k \end{bmatrix},
 \\ \nonumber  &\ \ \ \ \  \| \tilde{x}_{0|0} \|_{\infty} \leq \delta^x_0, \ \|w_i\|_{\infty} \leq \eta_w, \ \| v_j \|_{\infty} \leq \eta_v,
 \\ \nonumber &\ \ \ \ \ \forall i \in \{ 0,...,k-1 \}, \ \forall j \in \{ 0,...,k \}\yong{.} %, %\ and%\mathbb{A}^q_k \ is \ defined \ in \ Lemma \ref{lem:resdef}.  
  \end{align}
 % \vspace{-0.5cm}
 Again by matrix lower bound theorem, a sufficient condition  for the above inequality to hold is that $\exists \overline{d} \in \mathbb{R}$, such that: \vspace{-0.15cm}
  \begin{align} \label{eq:suff}
 % \vspace{-1.5cm}
   &\min \limits_{t_k,d_{0:k}}\|t'_k\|_2>\frac{{\delta}^{q,tri}_{r,k}}{\sigma_{min}\begin{bmatrix}  \mathbb{T}^{q,q'}_k & \mathbb{B}^{q,q'}_k & \mathbb{D}^{q,q'}_k  
 \end{bmatrix}}  
 \\[-0.2cm] \nonumber &s.t. \ t'_k=\begin{bmatrix} t^{\top}_k & u^{q'\top}_{0:k} & d^{q*\top}_{0:k} \end{bmatrix}^{\top},\|d^{q*}_{0:k}\|_2 \geq \overline{d},
 %t_k= \begin{bmatrix} \tilde{x}^{\top}_{0|0} & w^{\top}_0 & \dots & w^{\top}_{k-1} & v^{\top}_0 & \dots & v^{\top}_k \end{bmatrix},
 \\[-0.05cm] \nonumber &\ \ \ \ \ t_k= \begin{bmatrix} \tilde{x}^{\top}_{0|0} & w^{\top}_0 & \dots & w^{\top}_{k-1} & v^{\top}_0 & \dots & v^{\top}_k \end{bmatrix},
 \\ \nonumber  &\ \ \ \ \  \| \tilde{x}_{0|0} \|_{\infty} \leq \delta^x_0, \ \|w_i\|_{\infty} \leq \eta_w, \ \| v_j \|_{\infty} \leq \eta_v,
 \\ \nonumber &\ \ \ \ \ \forall i \in \{ 0,...,k-1 \}, \ \forall j \in \{ 0,...,k \}. %\ and%\mathbb{A}^q_k \ is \ defined \ in \ Lemma \ref{lem:resdef}
  \end{align}
 % \vspace{-0.5cm} 
  Finally, since ${\delta}^{q,tri}_{r,k} \leq \overline{\delta}^{q,tri}_{r,k}$ and
  \begin{footnotesize}
  \begin{align}
  \nonumber \|t'_k\|_2 =\| \begin{bmatrix} t^{\top}_k & u^{q'\top}_{0:k} & d^{q*\top}_{0:k} \end{bmatrix}\|_2 \geq \sqrt{0^2+0^2+\|d^{q*\top}_{0:k}\|^2_2}=\| d^{q*\top}_{0:k}\|_2,
  \end{align}
  \end{footnotesize}
  
  \vspace{-0.4cm}
\noindent   then a sufficient condition for \eqref{eq:suff} is that %$\exists \overline{d} \in \mathbb{R}$, such that: 
   \vspace{-0.15cm} \begin{align} \label{eq:suff2}
  \|d^{q*\top}_{0:k}\|_2 >\frac{\overline{\delta}^{q,tri}_{r,k}}{\sigma_{min}\moh{(}\begin{bmatrix}  \mathbb{T}^{q,q'}_k & \mathbb{B}^{q,q'}_k & \mathbb{D}^{q,q'}_k\end{bmatrix}\moh{)}}.%{\overline{\delta}^{q,tri}_{r,k}}\moh{/(}{\sigma_{min}\moh{(}\begin{bmatrix}  \mathbb{T}^{q,q'}_k & \mathbb{B}^{q,q'}_k & \mathbb{D}^{q,q'}_k\end{bmatrix}}\moh{))}.
  \end{align}
 
 \vspace{-0.1cm} 
Now suppose that $T^q_2 \neq T^{q'}_2$ (otherwise the matrix in the denominator of \eqref{eq:suff2} is zero and it never holds). Asymptotically speaking, the right hand side of \eqref{eq:suff2} converges to $\tilde{\delta} \triangleq \max \{0, (\overline{\delta}^{q,tri}_r/\overline{\sigma}^{q,q'}) \}$, since $\overline{\delta}^{q,tri}_{r,k}$ converges to $\overline{\delta}^{q,tri}_r$ and the least singular value in the denominator either diverges or converges to some steady value $\overline{\sigma}^{q,q'}$. So we set $\overline{d}$ equal to any real number strictly grater than $\tilde{\delta}$. By unlimited energy assumption for attack signal, after some large enough time step $K$, the monotone increasing function $\|d^{q*}_{0:k}\|_2$, exceeds $\overline{d}$ and so the system will be mode detectable. %In the second case, $\|d^{q*}_{0:k}\|_2$ never exceeds $\overline{d}$. In this case, the attack signal should converge to zero, since otherwise $\|d^{q*}_{0:k}\|_2$ diverges and can not be bounded above by $\overline{d}$. Hence in this case, there is \emph{no attack signal} asymptotically which is even stronger than asymptotic strong mode detectability. 
%This completes the proof of sufficiency of \ref{item:first}. Finally, sufficiency of \ref{item:third} directly follows from sufficiency of \ref{item:first} and \ref{item:second}.
\end{proof}

\end{document}